\documentclass[journal]{IEEEtran}

\usepackage{cite}
\usepackage{amsmath,amssymb,amsfonts}
\usepackage{bm}
\usepackage{graphicx}
\usepackage{xcolor}
\usepackage{algorithm}
\usepackage{algorithmic}
\usepackage[explicit]{titlesec}

\usepackage[caption=false,font=footnotesize,labelfont=rm,textfont=rm]{subfig}

\makeatletter
\renewcommand{\subsection}{%
	\@startsection{subsection}{2}{\z@}%
	{2.0ex plus 0.5ex minus 0.5ex}%
	{0.5ex plus 0.2ex minus 0ex}%
	{\normalfont\normalsize\itshape}}
\makeatother

\usepackage[colorlinks=true,citecolor=red,linkcolor=blue,urlcolor=blue]{hyperref}

\newtheorem{proposition}{Proposition}

\newtheorem{remark}{Remark}

\newenvironment{proof}
{\noindent\textit{Proof:}\ }
{\hfill$\square$\par\medskip}

\begin{document}
	
	
	\title{\fontsize{16 pt}{\baselineskip}\selectfont Intelligent Reflecting Surface Deployment for Low-Altitude Coverage: Illumination Geometry, Directional Characteristics, and Optimization}
	
	\author{Guoying~Zhang,
		Qingqing Wu, 
		Ailing Zheng,
		Xingxiang Peng,
		Wen Chen,
		Wei Feng
		\thanks{Guoying Zhang, Qingqing Wu, Ailing Zheng, Xingxiang Peng and Wen Chen are with the School of Integrated Circuits, Shanghai Jiao TongUniversity, Shanghai 200240, China (e-mail: \{gy\_zhang; qingqingwu; ailing.zheng; peng\_xingxiang; wenchen\}@sjtu.edu.cn). Wei Feng is with the Department of Electronic Engineering, State Key Laboratory of Space Network and Communications, Tsinghua University, Beijing 100084, China (e-mail: fengwei@tsinghua.edu.cn). }
	}
	
	\maketitle

 	\begin{abstract}
	Terrestrial base stations (BSs) are typically configured with fixed downtilt to serve ground users, resulting in weak illumination of low-altitude airspace even under line-of-sight (LoS) propagation. In this paper, we establish a channel model that incorporates BS and intelligent reflecting surface (IRS) radiation patterns for three-dimensional (3D) low-altitude coverage while preserving the existing BS configuration. We formulate a budget-constrained IRS deployment problem that jointly determines candidate-site selection, IRS orientations, and phase shifts to maximize the worst-case signal-to-noise ratio (SNR) over the 3D low-altitude airspace. The selected sites and optimized IRS parameters remain fixed after deployment, yielding a quasi-static IRS configuration. We characterize the illumination geometry between the fixed-downtilt BS and rooftop candidates by deriving the nonnegative installation-height range satisfying the BS main-lobe condition. The separation between the mapped main-lobe height boundaries grows linearly with horizontal BS-to-site distance and decreases inversely with the number of BS antennas.
	We further derive an analytical lower bound on the regional worst-case normalized array gain achievable through IRS phase design over served directions with different direction spans. The resulting sufficient direction span decreases inversely with the square root of the number of IRS elements when the same worst-case normalized gain guarantee is maintained. We develop a mixed-integer alternating optimization (AO) algorithm to solve the resulting problem. Simulation results validate the analytical characterizations and show that the proposed scheme achieves higher worst-case SNR than benchmarks across different deployment budgets.
 	\end{abstract}
 	

	\section{Introduction}\label{sec:intro}
	\bstctlcite{BSTcontrol}

	Low-altitude applications such as unmanned aerial vehicle (UAV) inspection, logistics, and urban air mobility require reliable connectivity throughout a three-dimensional (3D) region, over which propagation conditions and antenna gains vary with both horizontal position and altitude~\cite{shakhatreh2019uav,mozaffari2019tutorial,zeng2019accessing}. Cellular networks have been widely investigated for such aerial access~\cite{lin2018sky,fotouhi2019survey,geraci2018understanding}, and the Third Generation Partnership Project (3GPP) has examined their feasibility and limitations for aerial vehicles~\cite{muruganathan2018overview,3gpp36777}.
	As the horizontal position and altitude of an aerial terminal vary, the elevation angle relative to the BS changes over a broad range, and the received antenna gain follows different portions of the vertical radiation pattern. 
	Line-of-sight (LoS) propagation between the BS and the aerial terminal establishes an unobstructed path, while the antenna gain along that path depends on the position within the radiation pattern and can remain low. Aerial locations may be served through sidelobes or may lie near radiation nulls, which produces pronounced spatial variation in the received signal strength across the low-altitude region~\cite{euler2018mobility,geraci2018understanding}. The weakest locations can consequently limit the service quality of the entire region even when most locations retain LoS propagation.  
	
	Two transmit-side remedies have been investigated, each with a distinct cost. Reconfiguring the downtilt or installing uptilted antennas raises the gain toward aerial elevations~\cite{chowdhury2021connected,maeng2021uptilt}, but reduces the gain available to ground users for which the downtilt was selected and increases inter-cell interference. In addition, the BS main-lobe illumination covers only part of the airspace since the elevation angle relative to the BS varies with both horizontal position and altitude~\cite{nguyen2018coexistence,geraci2018understanding}. Densifying the BSs or introducing active relays and aerial platforms creates additional transmission points~\cite{alhourani2014optimal,mozaffari2015drone} at the cost of additional transmit power and radio-frequency (RF) hardware in networks already deployed for terrestrial service.
	These limitations motivate a passive coverage-enhancement method that preserves the existing BS transmission configuration and avoids the additional power consumption and RF hardware cost of new active transmission nodes.

	Intelligent reflecting surfaces (IRSs) enable configurable passive signal redirection. An IRS is a planar array of passive elements that reflect the incident signal with adjustable phase shifts, allowing the reflected fields of the elements to add constructively at a location~\cite{wu2019joint}. 
	An IRS redirects an incident signal without RF chains. This passive operation distinguishes an IRS from an active relay and reduces power consumption and hardware cost~\cite{wu2021tutorial,wu2020discrete}. IRS-aided aerial communication studies have focused on aerial terminals and air-ground links, where UAV trajectories and IRS phase shifts are jointly optimized to improve performance~\cite{li2020uav,you2021enabling}. IRS-assisted aerial links have been investigated under malicious jamming~\cite{ji2022trajectory,li2026openset}. The IRS-assisted designs in~\cite{li2020uav,you2021enabling,ji2022trajectory} adapt the reflection coefficients according to the location or trajectory of an aerial terminal. Three-dimensional passive coverage has been studied through aerial IRS placement and max-min passive beamforming with beam flattening over target areas~\cite{lu2021aerial,ma2022passive}. For a given IRS location and orientation, the quasi-static design in~\cite{jiang2025quasi} shapes the reflected beam over a prescribed angular range using fixed phase shifts shared across user locations.
 
	Beyond link-level and regional beam design, network-level studies incorporate IRS locations and deployment densities, along with association rules, into system analysis and optimization~\cite{wu2024intelligent}. Stochastic-geometry analyses characterize coverage and throughput under different IRS densities and association rules across cellular networks as well as hybrid active-passive and RF-powered Internet of Things networks~\cite{lyu2020spatial,lyu2021hybrid,sun2023rfpowered}. Other studies jointly optimize multi-IRS placement and power control for target-area service or coordinate IRS deployment with movable antenna positions for cost-effective coverage~\cite{ling2021placement,gao2026twoscale}. Network-level formulations study the effect of spatial deployment and association on network performance, with each IRS characterized by the corresponding location and associated link model. The network-level deployment variables comprise IRS locations, densities, and associations with prescribed panel mounting geometry and directional response.
	
	Practical deployment associates IRS panels with structures such as rooftops and walls in urban environments~\cite{etsi2023ris,cao2021aerial}. The design in~\cite{fu2024site} considers candidate building facets, site-dependent blockage, and the BS and IRS radiation patterns. An integer linear program selects IRS locations and user associations from link-level metrics at ground locations, while the panel orientation at each selected facet remains fixed. Flexible-reflector designs optimize continuous positions and rotation angles to maximize the minimum expected received power over a two-dimensional area~\cite{lu2025flexible}. The reflectors are passive metal surfaces whose design variables are positions and rotation angles within a continuous placement region. 
	The low-altitude design in~\cite{jiang2026lowaltitude} optimizes UAV placement, per-element rotations, and phase shifts for a single blocked BS-to-ground-terminal communication link. None of these designs combines budget-constrained site selection, orientation optimization, and phase-shift configuration so as to guarantee a worst-case signal-to-noise ratio (SNR) over a continuous 3D low-altitude region under a fixed BS downtilt.
	
	Motivated by the above, this paper studies quasi-static IRS deployment for low-altitude coverage under a fixed BS downtilt. 
	All configurations remain fixed during operation across all terminal locations, and the objective maximizes the worst-case SNR over the prescribed 3D target region. The main contributions of this paper are summarized as follows:
	
	\begin{figure}[t!]
		\centering
		\includegraphics[width=3.5 in]{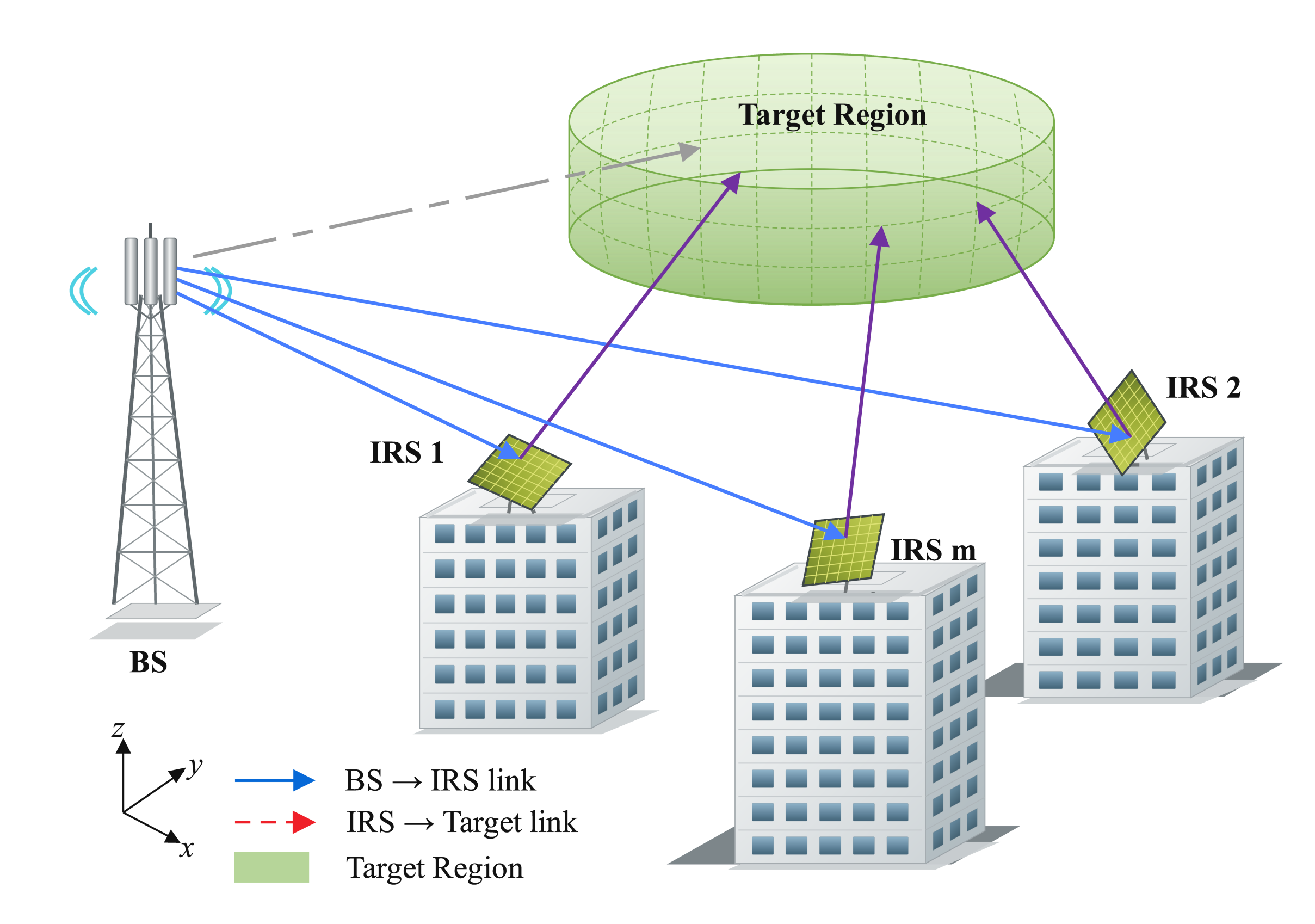}
		\caption{Quasi-static IRS-assisted low-altitude airspace coverage system.}
		\label{fig:scenario}
	\end{figure}
	
	\begin{itemize}
		
	\item We establish a radiation-pattern-aware channel model for 3D low-altitude IRS deployment under an existing fixed-downtilt BS. Based on this model, we formulate a budget-constrained deployment problem, where candidate-site selection, IRS orientations, and phase shifts are jointly optimized to maximize the worst-case SNR. Binary site selection and the nonlinear dependence of the SNR on IRS orientations and phase shifts make the resulting problem mixed-integer and nonconvex.
	
	\item We characterize candidate-site illumination geometry and  directional characteristics under different BS and IRS array
	sizes. The first characterization gives the nonnegative installation-height range satisfying the BS main-lobe condition. The separation between the corresponding height boundaries grows linearly with horizontal BS-to-site distance and decreases inversely with the number of BS antennas. The second gives an analytical lower bound on the maximum achievable worst-case normalized array gain over the IRS reflection directions. The resulting sufficient direction span decreases inversely with the square root of the number of IRS elements under the same worst-case normalized gain guarantee.
	
	\item The continuous-region problem contains infinitely many SNR constraints, precluding direct optimization. We therefore approximate the target region using uniform samples augmented at predicted heights associated with weak direct-link illumination near array-factor nulls. Based on the finite problem, we develop a mixed-integer alternating optimization (AO) algorithm that updates candidate-site selection through a mixed-integer linear program (MILP) and IRS orientations and phase shifts through concave quadratic local models. Continuous-block candidates are evaluated using the SNR objective, and accepted updates generate a nondecreasing convergent sequence of worst-case SNR values.
		
	\item Simulation results validate both analytical characterizations and confirm the convergence of the proposed mixed-integer AO algorithm. The proposed scheme achieves higher worst-case SNR than every considered benchmark across all tested deployment budgets. The comparisons further demonstrate the importance of IRS orientation optimization and phase shift design over the target region while accounting for the IRS element radiation pattern.
		
	\end{itemize}
	
	The rest of this paper is organized as follows. Section~\ref{sec:system_model} presents the system model and max-min SNR formulation, Section~\ref{sec:analysis} analyzes BS and IRS array-size effects on illumination
	geometry and directional characteristics, Section~\ref{sec:alg2} develops the mixed-integer AO algorithm, Section~\ref{sec:results} presents numerical results, and Section~\ref{sec:conclusion} concludes the paper.

	\emph{Notation:} Scalars, vectors, and matrices are denoted by italic, bold lowercase, and bold uppercase letters. $\|\cdot\|$ is the Euclidean norm, $\operatorname{diag}(\cdot)$ a diagonal matrix, $(\cdot)^T$ the transpose, $(\cdot)^*$ the conjugate, and $|\cdot|$ the absolute value of a scalar or the cardinality of a set. $\Re\{\cdot\}$ and $\Im\{\cdot\}$ denote the real and imaginary parts, $j=\sqrt{-1}$ is the imaginary unit, $\angle(\mathbf a,\mathbf b)$ denotes the angle between two vectors, and $\arg z$ denotes the phase of a complex scalar $z$. $\operatorname{col}(\cdot)$ stacks the argument vectors into one column vector, $\mathbf A\preceq\mathbf B$ means $\mathbf B-\mathbf A$ is positive semidefinite, $\operatorname{atan2}(y,x)$ is the two-argument arctangent, and $O(\cdot)$ and $\Theta(\cdot)$ are standard asymptotic notations.

	\section{System Model and Problem Formulation}\label{sec:system_model}
	This section describes the IRS deployment and establishes the channel model for the direct and IRS-assisted links over the target airspace. We then formulate a max-min SNR problem over prescribed 3D locations for joint optimization of site selection, IRS orientations, and phase shifts.

	\subsection{Deployment Configuration}\label{sec:config}
	As shown in Fig.~\ref{fig:scenario}, we consider an IRS-assisted downlink system covering a target low-altitude airspace $\mathcal V\subset\mathbb R^3$ above the building layer. We adopt a Cartesian coordinate system centered at the ground projection of the BS, with the positive $z$-axis upward and the $x$--$y$ plane on the ground. The BS is located at $\mathbf b_0=[0,0,H_B]^T$ and is equipped with a uniform linear array (ULA) of $N_t$ antennas with spacing $\lambda/2$ and fixed electrical downtilt $\theta_{\mathrm{tilt}}$, where $\lambda$ is the carrier wavelength. The BS transmits with power $P_0$, and the receiver noise power is $\sigma^2$. 	
	The area contains buildings $\mathcal B=\{1,\ldots,N_{\mathrm{bld}}\}$, where building $b\in\mathcal B$ is modeled as an axis-aligned cuboid with height $H_b$. Candidate IRS sites $\mathcal M_0=\{1,\ldots,M_0\}$ are located on rooftops with LoS-dominant links to the BS and target airspace. IRS $m$ is centered at $\mathbf q_m=[x_m,y_m,H_{b(m)}+h_{\mathrm{inst},m}]^T$ on building $b(m)$, where $h_{\mathrm{inst},m}$ denotes the installation height above the rooftop. The binary variable $s_m\in\{0,1\}$ indicates site selection subject to $\sum_{m\in\mathcal M_0}s_m\le M_{\max}$, where $M_{\max}$ is the maximum number of deployed IRSs. After deployment planning, the selected sites, panel orientations, and IRS phase shifts remain fixed throughout operation.
	
	For each candidate site $m\in\mathcal M_0$, the associated IRS is modeled as a uniform planar array (UPA) with element set $\mathcal N_m=\{1,\ldots,N_m\}$ and spacing $d_e=\lambda/2$, where $N_m=N_{m,h}N_{m,v}$ and $N_{m,h}$ and $N_{m,v}$ denote the horizontal and vertical element numbers, respectively. The orientation of IRS $m$ is specified by $\boldsymbol{\tau}_m=[\theta_m,\phi_m]^T$, where $\theta_m$ denotes the inclination of the panel normal from the positive $z$-axis and $\phi_m$ denotes the azimuth angle of the panel normal. At candidate site $m$, the feasible orientation set is
	\setlength\abovedisplayskip{1pt}
	\setlength\belowdisplayskip{1pt}
	\begin{equation}
		\!\!	\mathcal T_m
		\!=\!
		\left\{\!
		(\!\theta_m,\phi_m\!)\,\middle|\,
		\theta_m \in [\theta_m^{\min},\theta_m^{\max}],
		\phi_m \in [\phi_m^{\min},\phi_m^{\max}]
		\!\right\}.
		\label{eq:orientation_feasible_set}
	\end{equation}
	where $\theta_m^{\min}$ and $\theta_m^{\max}$ denote the inclination bounds, while $\phi_m^{\min}$ and $\phi_m^{\max}$ denote the azimuth bounds at candidate site $m$.
	The panel normal vector corresponding to $\boldsymbol{\tau}_m$ is
	\begin{equation}
		\mathbf{n}_m(\boldsymbol{\tau}_m)
		=
		[\sin\theta_m\cos\phi_m,\,
		\sin\theta_m\sin\phi_m,\,
		\cos\theta_m]^T .
		\label{eq:panel_normal}
	\end{equation}
	A local-to-global rotation matrix $\mathbf R_m(\boldsymbol{\tau}_m)\in\mathbb R^{3\times3}$ maps the local normal $\hat{\mathbf z}=[0,0,1]^T$ to $\mathbf n_m(\boldsymbol{\tau}_m)$ and is written as
	\begin{equation}
		\!\!	\mathbf R_m(\boldsymbol{\tau}_m)
		\!	=\!
		\begin{bmatrix}
			\cos\theta_m\cos\phi_m & \!-\sin\phi_m &\! \sin\theta_m\cos\phi_m\\
			\cos\theta_m\sin\phi_m &\! \cos\phi_m &\! \sin\theta_m\sin\phi_m\\
			-\sin\theta_m &\! 0 &\! \cos\theta_m
		\end{bmatrix}.
		\label{eq:alg2_rotation_matrix}
	\end{equation}
	The global coordinate of element $n$ on IRS $m$ is
	\begin{equation}
		\mathbf{p}_{m,n}
		=
		\mathbf{q}_m
		+
		\mathbf R_m(\boldsymbol{\tau}_m)\bar{\mathbf p}_{m,n},
		\label{eq:element_global_position}
	\end{equation}
	where $\bar{\mathbf p}_{m,n}$ is the local coordinate of element $n$ relative to the center of IRS $m$. Element $n$ is indexed by the row--column pair $(i_h,i_v)$ through $n=(i_v-1)N_{m,h}+i_h$, with $i_h\in\{1,\ldots,N_{m,h}\}$ and $i_v\in\{1,\ldots,N_{m,v}\}$, and $\bar{\mathbf p}_{m,n}$ satisfies
	\begin{equation}
		\bar{\mathbf p}_{m,n}
		=
		d_e
		\left[
		i_h-\frac{N_{m,h}+1}{2},
		i_v-\frac{N_{m,v}+1}{2},
		0
		\right]^T.
	\end{equation}
	
	\subsection{Propagation Model}
	\label{sec:channel}
 	
	We consider frequency-flat baseband channels for the direct BS-to-location link and the IRS-assisted link. Fig.~\ref{fig:geo} illustrates the direct link and the IRS-assisted link through IRS $m$ for an arbitrary target-airspace location $\mathbf q_u\in\mathcal V$ indexed by $u$. The direct BS-to-location, BS-to-IRS, and IRS-to-location links are modeled in the far field given their propagation distances.
	
	\begin{figure}[t!]
		\centering
		\includegraphics[width=3.5in]{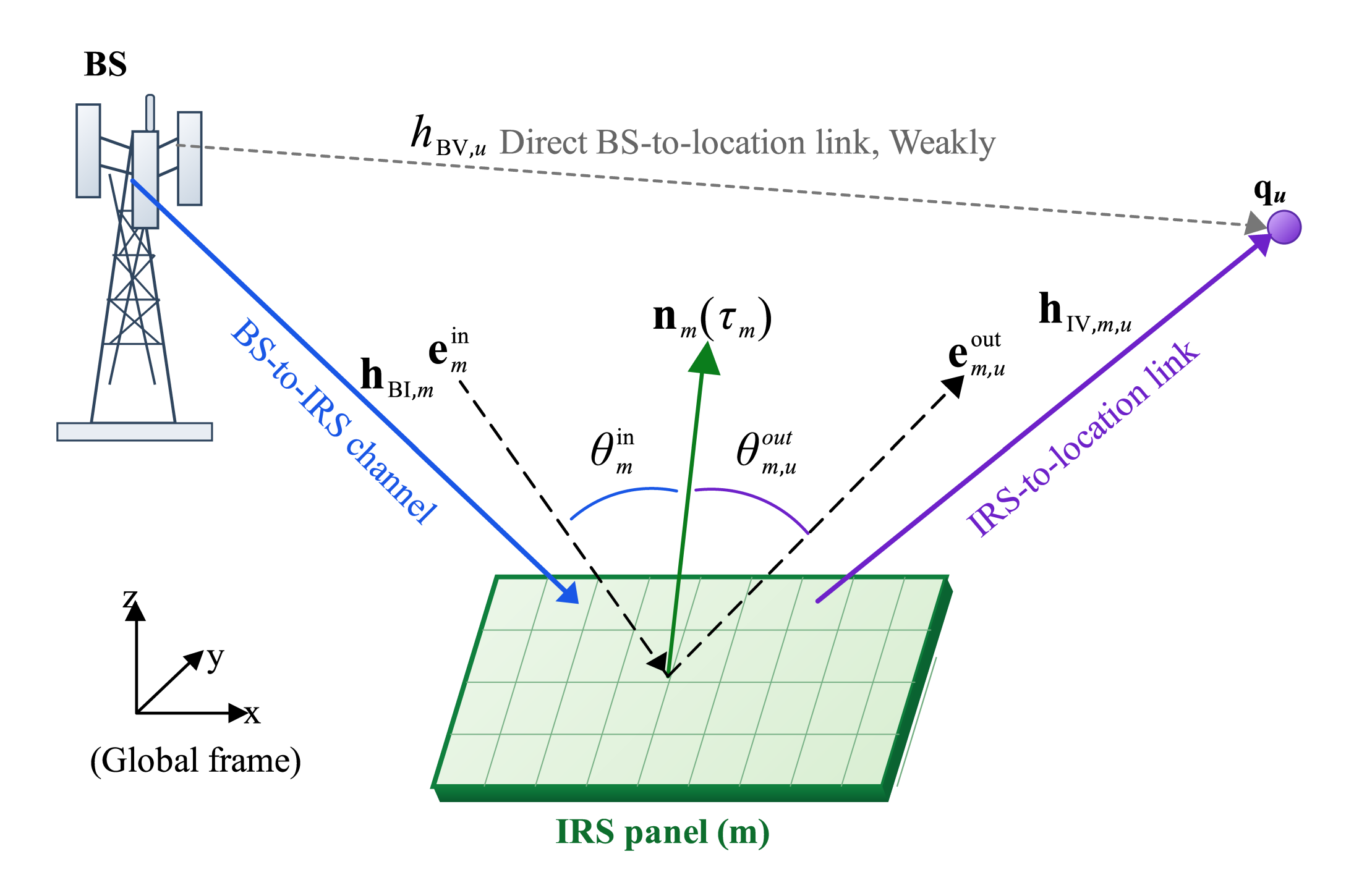}
		\caption{Geometry of the direct and IRS-assisted links.}
		\label{fig:geo}
	\end{figure}
	
	\subsubsection{BS Antenna Pattern and IRS Element Pattern}
	
	The BS antenna gain follows the 3GPP model~\cite{3gpp38901}. 
	For a location $\mathbf r=[x,y,z]^T$, the elevation angle from the BS to $\mathbf r$ is $\theta_{\mathrm B}(\mathbf r)=\operatorname{atan2}(z-H_B,\sqrt{x^2+y^2})$, and the BS antenna gain toward $\mathbf r$ can be expressed as
	\begin{equation}
		G_{\mathrm B}(\mathbf r)
		=
		N_t
		G_{\mathrm e}\bigl(\theta_{\mathrm B}(\mathbf r)\bigr)
		A_{\mathrm F}\bigl(\theta_{\mathrm B}(\mathbf r)\bigr).
		\label{eq:bs_pattern}
	\end{equation}
	Here $N_t$ is the number of BS antennas, $G_{\mathrm e}(\theta)=G_{\mathrm e}^{\max}10^{-\frac{1}{10}\min\{12(\theta/\theta_{\mathrm{3dB}})^2,\mathrm{SLA}_V\}}$ is the 3GPP element gain at elevation angle $\theta$, $G_{\mathrm e}^{\max}$ is the element boresight gain, $\theta_{\mathrm{3dB}}$ is the half-power beamwidth, $\mathrm{SLA}_V$ is the side-lobe attenuation, and $A_{\mathrm F}(\theta)=|F_{\mathrm B}(\theta)|^2$ denotes the array factor. Specifically, $F_{\mathrm B}(\theta)$ can be expressed as
	\begin{equation}
		F_{\mathrm B}(\theta)
		\triangleq
		\frac{1}{N_t}\sum_{n=0}^{N_t-1}e^{j\frac{2\pi}{\lambda}d_v n(\sin\theta+\sin\theta_{\mathrm{tilt}})},
		\label{eq:bs_array_response}
	\end{equation}
	where $d_v=\lambda/2$ is the vertical antenna spacing~\cite{balanis2016antenna}.
	
	The IRS element radiation pattern (ERP) is modeled by a front-side cosine-power directional gain pattern~\cite{jiang2026lowaltitude,balanis2016antenna}. The IRS ERP gain is given by
	\begin{equation}
		G_{\mathrm I}(\omega)
		=
		\begin{cases}
			G_{\mathrm I}^{\max}\cos^{p_I}\omega,
			& 0\le\omega<\frac{\pi}{2},\\
			0,
			& \frac{\pi}{2}\le\omega\le\pi,
		\end{cases}
		\label{eq:irs_erp}
	\end{equation}
	where $\omega$ is measured from the IRS normal $\mathbf n_m(\boldsymbol{\tau}_m)$, $p_I\ge2$ is the directivity exponent, and $G_{\mathrm I}^{\max}=2(p_I+1)$ is the broadside element gain. For an arbitrary airspace location $\mathbf q_u\in\mathcal V$, the incident and outgoing angles of IRS $m$ are defined as
	\begin{align}
		\theta_m^{\mathrm{in}}
		&=
		\angle\bigl(
		\mathbf b_0-\mathbf q_m,\,
		\mathbf n_m(\boldsymbol{\tau}_m)
		\bigr),
		\label{eq:angle_in}\\
		\theta_{m,u}^{\mathrm{out}}
		&=
		\angle\bigl(
		\mathbf q_u-\mathbf q_m,\,
		\mathbf n_m(\boldsymbol{\tau}_m)
		\bigr).
		\label{eq:angle_out}
	\end{align}
	The incident and outgoing ERP gains used in the cascaded channel are $G_{\mathrm I}(\theta_m^{\mathrm{in}})$ and $G_{\mathrm I}(\theta_{m,u}^{\mathrm{out}})$, respectively.
	
	\subsubsection{Channel Model}
	Target airspace and rooftop IRS locations favor LoS-dominant propagation along the direct and IRS-assisted paths. Accordingly, the direct BS-to-location, BS-to-IRS, and IRS-to-location links are modeled as LoS channels. The large-scale attenuation of each link is modeled as~\cite{rappaport2002wireless}
	\begin{equation}
		\!\!	\beta_X
		\! 	=\!
		C_{\mathrm{L}}\!\!
		\left(\frac{d_X}{d_0}\right)^{-\alpha_{\mathrm{L}}}\!\!,
		X \!\!\in\!\! \{(\mathrm{BV},u),  (\mathrm{BI},m), (\mathrm{IV},m,u)\},
		\label{eq:pathloss_compact}\! \!
	\end{equation}
	where $\alpha_{\mathrm{L}}$ is the path-loss exponent, $d_X$ is the corresponding link distance, and $C_{\mathrm{L}}=\left(\frac{\lambda}{4\pi d_0}\right)^2$ is the free-space gain at the reference distance $d_0$.  
	Let $d_{\mathrm{BV},u}=\|\mathbf{q}_u-\mathbf{b}_0\|$ denote the BS-to-$u$ distance, $d_{\mathrm{BI},m}=\|\mathbf{q}_m-\mathbf{b}_0\|$ the BS-to-IRS distance, and $d_{\mathrm{IV},m,u}=\|\mathbf{q}_u-\mathbf{q}_m\|$ the IRS-to-$u$ distance.\footnote{The BS-to-location, BS-to-IRS, and IRS-to-location links satisfy the corresponding Fraunhofer far-field conditions. In particular, the BS-to-IRS and IRS-to-location distances exceed the IRS Fraunhofer distance $d_{\mathrm F,m}=2D_{\mathrm{ap},m}^2/\lambda$, where $D_{\mathrm{ap},m}$ is the aperture diagonal.}
	The direct BS-to-$u$ channel is modeled as
	\begin{equation}
		\!\!\!\!	h_{\mathrm{BV},u}
		\!	=\!
		\sqrt{
			\beta_{\mathrm{BV},u}
			N_t
			G_{\mathrm e}\bigl(\theta_{\mathrm B}(\mathbf q_u)\bigr)
		}
		F_{\mathrm B}\bigl(\theta_{\mathrm B}(\mathbf q_u)\bigr)
		e^{-j\frac{2\pi}{\lambda}d_{\mathrm{BV},u}}.
		\label{eq:direct_channel}\!\!
	\end{equation}
	For the reflected link through the IRS at candidate site $m$, let $\mathbf e_m^{\mathrm{in}}\triangleq(\mathbf b_0-\mathbf q_m)/\|\mathbf b_0-\mathbf q_m\|$ denote the unit vector from IRS $m$ to the BS and $\mathbf e_{m,u}^{\mathrm{out}}\triangleq(\mathbf q_u-\mathbf q_m)/\|\mathbf q_u-\mathbf q_m\|$ the unit vector from IRS $m$ to $u$. The array response of IRS $m$ in the direction $\mathbf e_m^{\mathrm{in}}$ and the array response in the direction $\mathbf e_{m,u}^{\mathrm{out}}$ are
	\begin{align}
		[\mathbf a_m^{\mathrm{in}}]_n
		&=
		\exp\left(
		j\frac{2\pi}{\lambda}
		(\mathbf e_m^{\mathrm{in}})^T
		\mathbf R_m(\boldsymbol{\tau}_m)
		\bar{\mathbf p}_{m,n}
		\right),
		\label{eq:incident_response}\\
		[\mathbf a_{m,u}^{\mathrm{out}}]_n
		&=
		\exp\left(
		j\frac{2\pi}{\lambda}
		(\mathbf e_{m,u}^{\mathrm{out}})^T
		\mathbf R_m(\boldsymbol{\tau}_m)
		\bar{\mathbf p}_{m,n}
		\right).
		\label{eq:departure_response}
	\end{align}
	The BS-to-IRS channel $\mathbf h_{\mathrm{BI},m}$ and the IRS-to-$u$ channel $\mathbf h_{\mathrm{IV},m,u}$ are
	\begin{align}
		\mathbf h_{\mathrm{BI},m}
		&=
		\sqrt{
			\beta_{\mathrm{BI},m}
			N_t
			G_{\mathrm e}\bigl(\theta_{\mathrm B}(\mathbf q_m)\bigr)
			G_{\mathrm I}(\theta_m^{\mathrm{in}})
		}
		\nonumber\\
		&\quad\times
		F_{\mathrm B}\bigl(\theta_{\mathrm B}(\mathbf q_m)\bigr)
		e^{-j\frac{2\pi}{\lambda}d_{\mathrm{BI},m}}
		\mathbf a_m^{\mathrm{in}},
		\label{eq:incident_channel}\\
		\mathbf h_{\mathrm{IV},m,u}
		&=
		\sqrt{
			\beta_{\mathrm{IV},m,u}
			G_{\mathrm I}(\theta_{m,u}^{\mathrm{out}})
		}
		e^{-j\frac{2\pi}{\lambda}d_{\mathrm{IV},m,u}}
		\mathbf a_{m,u}^{\mathrm{out}}.
		\label{eq:departure_channel}
	\end{align}
	For each candidate site $m\in\mathcal M_0$, let 
	\begin{equation}
		\boldsymbol{\Phi}_m=\mathrm{diag}(e^{j\varphi_{m,1}},\ldots,e^{j\varphi_{m,N_m}}),
	\end{equation}
	where $\boldsymbol{\varphi}_m=[\varphi_{m,1},\ldots,\varphi_{m,N_m}]^T$ denotes the IRS phase shifts of IRS $m$ and $\varphi_{m,n}\in[0,2\pi)$ is the phase shift of the $n$-th reflecting element under unit reflection amplitude. The cascaded BS-IRS-$u$ channel through candidate $m$ is
	\begin{equation}
		h_{m,u}(\boldsymbol{\tau}_m,\boldsymbol{\varphi}_m)
		=
		\mathbf h_{\mathrm{IV},m,u}^T
		\boldsymbol{\Phi}_m
		\mathbf h_{\mathrm{BI},m}.
		\label{eq:cascaded_channel}
	\end{equation}
	Substituting~\eqref{eq:incident_channel} and~\eqref{eq:departure_channel} into~\eqref{eq:cascaded_channel}, the cascaded channel can be expressed as
	\begin{equation}
			h_{m,u}(\boldsymbol{\tau}_m,\boldsymbol{\varphi}_m)
		=
		\zeta_{m,u}(\boldsymbol{\tau}_m)
		\,
		e^{-j\psi_{m,u}}
		\,
		(\mathbf a_{m,u}^{\mathrm{out}})^T
		\boldsymbol{\Phi}_m
		\mathbf a_m^{\mathrm{in}} ,
		\label{eq:cascaded_channel_expand}
	\end{equation}
	where
	\begin{align}
		\zeta_{m,u}(\boldsymbol{\tau}_m)
		&\triangleq
		(
		\Lambda_{m,u}
		G_{\mathrm I}(\theta_m^{\mathrm{in}})
		G_{\mathrm I}(\theta_{m,u}^{\mathrm{out}}))^{1/2}, \\
		\psi_{m,u}
		&\triangleq
		\frac{2\pi}{\lambda}
		\left(
		d_{\mathrm{BI},m}+d_{\mathrm{IV},m,u}
		\right)
		-
		\vartheta_{\mathrm B,m},
	\end{align}	and $\vartheta_{\mathrm B,m}\triangleq\arg F_{\mathrm B}\bigl(\theta_{\mathrm B}(\mathbf q_m)\bigr)$ denotes the phase of the BS array response toward candidate site $m$, with $\Lambda_{m,u}\triangleq\beta_{\mathrm{BI},m}\beta_{\mathrm{IV},m,u}G_{\mathrm B}(\mathbf q_m)$.
	
	Define $\mathbf s\triangleq[s_1,\ldots,s_{M_0}]^T$, $\boldsymbol{\tau}\triangleq\{\boldsymbol{\tau}_m\}_{m\in\mathcal M_0}$, and $\boldsymbol{\varphi}\triangleq\{\boldsymbol{\varphi}_m\}_{m\in\mathcal M_0}$. The received signal at location $u$ is
	\begin{equation}
	\!\!	y_u
	\!	=\!
		\sqrt{P_0}
		\left(
		h_{\mathrm{BV},u}
	\!	+\!
				\sum\nolimits_{m\in\mathcal M_0}
		s_m h_{m,u}(\boldsymbol{\tau}_m,\boldsymbol{\varphi}_m)
		\right)x
	\!	+\!
		n_u,\!\!
		\label{eq:received_signal}
	\end{equation}
	where $x$ is the transmitted symbol with $\mathbb E[|x|^2]=1$, and $n_u\sim\mathcal{CN}(0,\sigma^2)$ denotes the additive noise at point $u$. The resulting SNR is
	\begin{equation}
	\!\!\!\!	\gamma_u(\mathbf s,\boldsymbol{\tau},\boldsymbol{\varphi})
	\!	=\!
		\frac{P_0}{\sigma^2}
		\left|
		h_{\mathrm{BV},u}
		\!\!+\!\!
			\sum\nolimits_{m\in\mathcal M_0}
	s_m h_{m,u}(\boldsymbol{\tau}_m,\boldsymbol{\varphi}_m)
		\right|^2.
		\label{eq:snr}\!\! 
	\end{equation}
	
	\subsection{Problem Formulation}\label{sec:formulation}
	Since the target airspace $\mathcal V$ contains infinitely many locations, we approximate the continuous-region minimum SNR over a finite location set for numerical optimization. 
	We first uniformly sample $\mathcal V$ at intervals of $\Delta_v$ to generate the basic locations. The uniform grid is augmented with locations of weak direct-link power near the nulls of the fixed-downtilt BS array factor. The array factor in~\eqref{eq:bs_array_response} identifies these locations. With $d_v=\lambda/2$, the nulls satisfy $\nu=2\xi/N_t$, where $\nu=\sin\theta+\sin\theta_{\mathrm{tilt}}$ and $\xi$ is an integer with $\xi\bmod N_t\neq0$. The corresponding null elevation angles are $\theta_\xi=\arcsin(2\xi/N_t-\sin\theta_{\mathrm{tilt}})$ for $\left|2\xi/N_t-\sin\theta_{\mathrm{tilt}}\right|\leq1$.
	At horizontal distance $\rho$ from the BS, the corresponding null height is $z_\xi(\rho)=H_B+\rho\tan\theta_\xi$.
	At each basic horizontal coordinate, all valid null heights $z_\xi(\rho)$ are included with repeated locations counted once. Let $\mathcal Q_{\mathrm{uni}}$ and $\mathcal Q_{\mathrm{null}}$ denote the uniformly sampled and null-height locations, respectively. The resulting finite location set is
	\begin{equation}
		\mathcal Q
		\triangleq
		\mathcal Q_{\mathrm{uni}}
		\cup
		\mathcal Q_{\mathrm{null}}
		\subset\mathcal V.
		\label{eq:finite_location_set}
	\end{equation}
	
	For the location set $\mathcal Q$, the minimum SNR is defined as
	\begin{equation}
		\gamma_{\min}
		(\mathbf s,\boldsymbol{\tau},\boldsymbol{\varphi})
		\triangleq
		\min_{\mathbf q_u\in\mathcal Q}
		\gamma_u
		(\mathbf s,\boldsymbol{\tau},\boldsymbol{\varphi}).
		\label{eq:minimum_snr_coverage_set}
	\end{equation}
	Introducing an auxiliary variable $\Gamma$, the joint rooftop deployment and IRS configuration problem is formulated as
	\begin{subequations}
		\label{prob:maxmin_snr_voxel}
		\begin{align}
			\textrm{(P1)}\ 
			\max_{\Gamma,\mathbf s,\boldsymbol{\tau},
				\boldsymbol{\varphi}}
			\ 
			& \Gamma
			\label{prob:maxmin_snr_obj}
			\\
			\mathrm{s.t.}\ 
			& \gamma_u
			(\mathbf s,\boldsymbol{\tau},\boldsymbol{\varphi})
			\geq\Gamma,
			\ 
			\forall \mathbf q_u\in\mathcal Q,
			\label{prob:maxmin_snr_voxel_region}
			\\
			& \sum_{m\in\mathcal M_0}s_m
			\leq M_{\max},
			\label{prob:maxmin_raw_budget}
			\\
			& s_m\in\{0,1\},
			\ 
			\forall m\in\mathcal M_0,
			\label{prob:maxmin_raw_binary}
			\\
			& \boldsymbol{\tau}_m\in\mathcal T_m,
			\ 
			\forall m\in\mathcal M_0,
			\label{prob:maxmin_raw_orientation}
			\\
			& 0\leq\varphi_{m,n}<2\pi,
			\ 
			\forall m\in\mathcal M_0,
			\ 
			\forall n\in\mathcal N_m.
			\label{prob:maxmin_raw_phase}
		\end{align}
	\end{subequations}
	Constraint~\eqref{prob:maxmin_snr_voxel_region} imposes the common SNR lower bound $\Gamma$ over $\mathcal Q$. Constraints~\eqref{prob:maxmin_raw_budget}--\eqref{prob:maxmin_raw_phase} specify the deployment budget, binary site selection, IRS orientations, and phase shifts. Problem~\textrm{(P1)} is mixed-integer and non-convex because the deployment variables are binary and the SNR in~\eqref{eq:snr} couples the site selection with the orientation- and phase-dependent cascaded channels.
	
	\begin{remark}
		The sampling interval $\Delta_v$ controls the spatial resolution of the finite location set and the number of SNR constraints. Let $\mathcal Q_{\mathrm{uni}}$ denote the centers of the Cartesian cells with side length $\Delta_v$ used to partition $\mathcal V$. The distance from any location in $\mathcal V$ to the nearest uniformly sampled location is at most $\sqrt{3}\Delta_v/2$, and the same bound holds for $\mathcal Q$ since $\mathcal Q_{\mathrm{uni}}\subseteq\mathcal Q$. For two horizontal positions $\mathbf p$ and $\mathbf p'$, the corresponding null heights satisfy $|z_\xi(\|\mathbf p\|)-z_\xi(\|\mathbf p'\|)|\leq|\tan\theta_\xi|\|\mathbf p-\mathbf p'\|$. Since every valid null height is inserted at each sampled horizontal coordinate, an interior point on the null-height surface indexed by $\xi$ lies within $\Delta_v/(\sqrt{2}|\cos\theta_\xi|)$ of an augmented location.
		For a fixed BS array, the uniform samples scale as $O(\Delta_v^{-3})$, whereas the null-height augmentation adds $O(\Delta_v^{-2})$ locations. The augmented locations capture the predicted weak direct-link regions, and both spatial approximation bounds vanish as $\Delta_v\to0$. Reducing $\Delta_v$ therefore improves both approximations, while the uniform grid remains the dominant source of sample growth.
	\end{remark}

	\section{BS and IRS Array-Size Analysis}\label{sec:analysis}
	The cascaded-link power in~\eqref{eq:cascaded_channel_expand} contains the BS-to-IRS illumination factor $\beta_{\mathrm{BI},m}G_{\mathrm B}(\mathbf q_m)$ and the IRS array-gain term $\left|(\mathbf a_{m,u}^{\mathrm{out}})^T\boldsymbol{\Phi}_m\mathbf a_m^{\mathrm{in}}\right|^2$. 
	
	\subsection{BS Illumination at Rooftop Candidates}
	\label{subsec:bs_illumination}
	We first characterize the BS main-lobe elevation interval and map the angular boundaries to the installation-height range of each rooftop candidate. For a prescribed array-factor loss threshold $\eta_B>0$, let $[\theta_{\mathrm{low}},\theta_{\mathrm{high}}]$ denote the connected elevation-angle interval containing $-\theta_{\mathrm{tilt}}$ and satisfying $10\log_{10}A_{\mathrm F}(\theta)\geq-\eta_B$. The component containing $-\theta_{\mathrm{tilt}}$ identifies the main-lobe interval among connected intervals satisfying the prescribed loss threshold.
	With $d_v=\lambda/2$, substituting $\nu=\sin\theta+\sin\theta_{\mathrm{tilt}}$ into~\eqref{eq:bs_array_response} and evaluating the geometric sum gives $\widetilde A_{\mathrm F}(\nu)\triangleq\left|\frac{\sin(N_t\pi\nu/2)}{N_t\sin(\pi\nu/2)}\right|^2$, where $\widetilde A_{\mathrm F}(0)=1$.
	The prescribed loss threshold is equivalent to $\widetilde A_{\mathrm F}(\nu)\geq10^{-\eta_B/10}$. Let $\nu_{\eta_B}\in(0,2/N_t)$ denote the unique positive main-lobe boundary satisfying $\widetilde A_{\mathrm F}(\nu_{\eta_B})=10^{-\eta_B/10}$. The symmetry of $\widetilde A_{\mathrm F}(\nu)$ around $\nu=0$ gives the main-lobe interval $\nu\in[-\nu_{\eta_B},\nu_{\eta_B}]$.
	When $\nu_{\eta_B}<1-\sin\theta_{\mathrm{tilt}}$, the corresponding elevation-angle boundaries are
	$\theta_{\mathrm{low}} = \arcsin(-\sin\theta_{\mathrm{tilt}}-\nu_{\eta_B})$
	and $\theta_{\mathrm{high}} = \arcsin(-\sin\theta_{\mathrm{tilt}}+\nu_{\eta_B})$.
	
	For candidate site $m\in\mathcal M_0$, define the horizontal distance from the BS as $\rho_m\triangleq\sqrt{x_m^2+y_m^2}$ and the BS elevation angle as $\theta_{\mathrm B,m}\triangleq\theta_{\mathrm B}(\mathbf q_m)$, which gives the panel-center height $q_{m,z}=H_{b(m)}+h_{\mathrm{inst},m}=H_B+\rho_m\tan\theta_{\mathrm B,m}$.
	The main-lobe condition $\theta_{\mathrm B,m}\in[\theta_{\mathrm{low}},\theta_{\mathrm{high}}]$ therefore requires the panel-center height $H_{b(m)}+h_{\mathrm{inst},m}$ to lie between $H_B+\rho_m\tan\theta_{\mathrm{low}}$ and $H_B+\rho_m\tan\theta_{\mathrm{high}}$.
	Subtracting $H_{b(m)}$ from the two panel-center height boundaries gives the installation-height boundaries $h_m^{\mathrm{low}}=H_B+\rho_m\tan\theta_{\mathrm{low}}-H_{b(m)}$ and $h_m^{\mathrm{high}}=H_B+\rho_m\tan\theta_{\mathrm{high}}-H_{b(m)}$ at candidate site $m$.
	The nonnegative installation-height range $\mathcal H_m$ satisfying the BS main-lobe condition is empty when $h_m^{\mathrm{high}}<0$ and equals $\mathcal H_m=[\max\{h_m^{\mathrm{low}},0\},h_m^{\mathrm{high}}]$ when $h_m^{\mathrm{high}}\geq0$.
	Candidates with $\mathcal H_m\neq\varnothing$ form the retained set $\mathcal M\subseteq\mathcal M_0$ for subsequent site selection and IRS configuration. Let $M\triangleq|\mathcal M|$ denote the number of retained candidates.
	For candidate site $m$, the separation between the two panel-center height boundaries corresponding to the prescribed array-factor loss threshold is
	\begin{equation}
		W_{\eta_B}(\rho_m)
		=
		\rho_m
		\left(
		\tan\theta_{\mathrm{high}}
		-
		\tan\theta_{\mathrm{low}}
		\right).
		\label{eq:mainlobe_band_width}
	\end{equation}
	The building height shifts both installation-height boundaries equally and affects whether $\mathcal H_m$ is empty while $W_{\eta_B}(\rho_m)$ remains independent of $H_{b(m)}$ and linear in $\rho_m$. The dependence on $N_t$ enters through the main-lobe boundary $\nu_{\eta_B}$ and is characterized by the large-array result below. 
	For the prescribed loss threshold, let $\mu_{\eta_B}\in(0,2)$ denote the unique solution of $\left[\sin(\pi\mu/2)/(\pi\mu/2)\right]^2=10^{-\eta_B/10}$. The following proposition characterizes the large-array scaling of the height-boundary separation $W_{\eta_B}(\rho_m)$ with respect to $N_t$. 
	
	\begin{proposition}
		\label{prop:mainlobe_band}
		For fixed $\eta_B$ and $\theta_{\mathrm{tilt}}$ with $\rho_m>0$, the height-boundary separation $W_{\eta_B}(\rho_m)$ defined in~\eqref{eq:mainlobe_band_width} follows the large-array asymptotic relation
		\begin{equation}
			W_{\eta_B}(\rho_m)
			\underset{N_t\to\infty}{\sim}
			\frac{2\mu_{\eta_B}\rho_m}
			{N_t\cos^3\theta_{\mathrm{tilt}}}.
			\label{eq:mainlobe_band_width_scaling}
		\end{equation}
		The relation yields $W_{\eta_B}(\rho_m)=\Theta(\rho_m/N_t)$.
	\end{proposition}
		
	\begin{proof}
		Set $\mu_{\eta_B,N_t}\triangleq N_t\nu_{\eta_B}\in(0,2)$ and define  $\widehat A_{N_t}(\mu)\triangleq \left| \sin(\pi\mu/2)/
		\left[N_t\sin\left(\pi\mu/(2N_t)\right)\right] \right|^2$ and
		$\widehat A_{\infty}(\mu)\triangleq \left[\sin(\pi\mu/2)/(\pi\mu/2)\right]^2$.
		The finite-array boundary $\mu_{\eta_B,N_t}$ and limiting boundary $\mu_{\eta_B}$ satisfy
		$\widehat A_{N_t}(\mu_{\eta_B,N_t}) = \widehat A_{\infty}(\mu_{\eta_B}) = 10^{-\eta_B/10}$.
		For $\mu$ in a closed neighborhood of $\mu_{\eta_B}$, the expansion
		$N_t\sin\left(\pi\mu/(2N_t)\right) = \pi\mu/2+O(N_t^{-2})$ 	gives
		$\widehat A_{N_t}(\mu) 	= 	\widehat A_{\infty}(\mu)+O(N_t^{-2})$
		uniformly as $N_t\to\infty$. Choose any sufficiently small $\epsilon>0$ such that $[\mu_{\eta_B}-\epsilon,\mu_{\eta_B}+\epsilon]\subset(0,2)$ and recall that $\widehat A_{\infty}(\mu)$ is strictly decreasing on $(0,2)$. The strict decrease gives $\widehat A_{\infty}(\mu_{\eta_B}-\epsilon)>10^{-\eta_B/10}>\widehat A_{\infty}(\mu_{\eta_B}+\epsilon)$ and the uniform approximation preserves both inequalities for sufficiently large $N_t$. The uniqueness of the finite-array boundary then gives $\mu_{\eta_B,N_t}\to\mu_{\eta_B}$ as $N_t\to\infty$.
		
		The uniform approximation further gives 
		$\widehat A_{\infty}(\mu_{\eta_B,N_t}) - \widehat A_{\infty}(\mu_{\eta_B}) 	= O(N_t^{-2})$.
		Since $\widehat A_{\infty}'(\mu_{\eta_B})\neq0$, the mean value theorem yields
		$\mu_{\eta_B,N_t} 	= \mu_{\eta_B}+O(N_t^{-2})$ and hence $\nu_{\eta_B}  = 	\mu_{\eta_B}/N_t+O(N_t^{-3})$.	
		To map the boundary expansion to the height separation, define
		$g_{\mathrm h}(\nu) \triangleq \tan[\arcsin(\nu-\sin\theta_{\mathrm{tilt}})]$.
		The elevation-angle boundaries and~\eqref{eq:mainlobe_band_width} then give
		$W_{\eta_B}(\rho_m) = \rho_m[ g_{\mathrm h}(\nu_{\eta_B})  - g_{\mathrm h}(-\nu_{\eta_B})]$.
		Since $g_{\mathrm h}'(0)=\cos^{-3}\theta_{\mathrm{tilt}}$, the central Taylor expansion gives
		$g_{\mathrm h}(\nu)-g_{\mathrm h}(-\nu) = 2\nu\cos^{-3}\theta_{\mathrm{tilt}}+O(\nu^3)$.
		Substituting $\nu_{\eta_B} = \mu_{\eta_B}/N_t+O(N_t^{-3})$ gives
		$W_{\eta_B}(\rho_m) = \frac{2\mu_{\eta_B}\rho_m} {N_t\cos^3\theta_{\mathrm{tilt}}} + O(\rho_mN_t^{-3})$,
		which establishes~\eqref{eq:mainlobe_band_width_scaling} and $W_{\eta_B}(\rho_m)=\Theta(\rho_m/N_t)$. 
	\end{proof}

	\begin{figure}[t!]
		\centering
		\includegraphics[width=3.3 in]{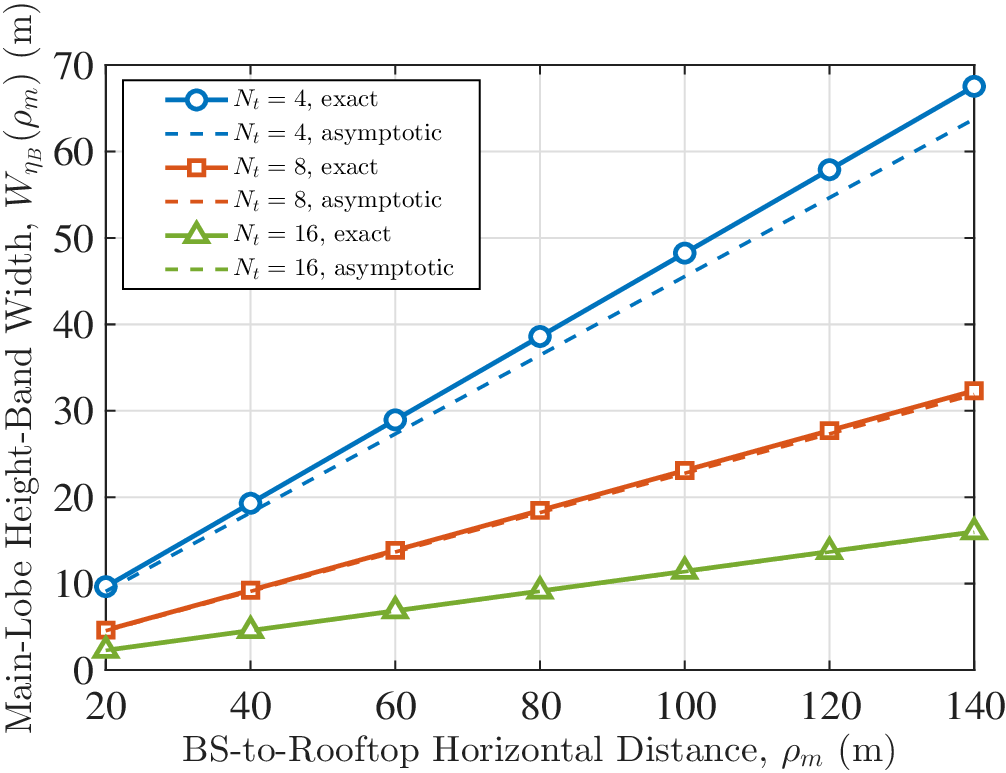}
		\caption{Main-lobe height-boundary separation versus horizontal distance.}
		\label{fig:mainlobe_band_validation}
	\end{figure}
	Proposition~\ref{prop:mainlobe_band} shows that the main-lobe height-boundary separation grows linearly with $\rho_m$ and decreases at an inverse-$N_t$ rate as the BS array grows. The two height boundaries therefore converge toward the downtilt-aligned panel-center height as $N_t$ increases.
	For each retained candidate $m\in\mathcal M$, the BS downtilt direction gives the panel-center height $z_{\mathrm c}(\rho_m)=H_B-\rho_m\tan\theta_{\mathrm{tilt}}$ and the corresponding rooftop installation height $h_{\mathrm c,m}=z_{\mathrm c}(\rho_m)-H_{b(m)}$. Since $\theta_{\mathrm{low}}<-\theta_{\mathrm{tilt}}<\theta_{\mathrm{high}}$ and the tangent function is increasing over the considered elevation range, the resulting heights satisfy $h_m^{\mathrm{low}}<h_{\mathrm c,m}<h_m^{\mathrm{high}}$. The nonnegative installation height in $\mathcal H_m$ closest to the BS downtilt direction is therefore
	\begin{equation}
		h_{\mathrm{inst},m}
		=
		\max\left\{h_{\mathrm c,m},0\right\}.
		\label{eq:height_projection_rule}
	\end{equation}
	After the installation height is fixed, the BS-to-IRS illumination of each retained candidate depends on the propagation loss and the BS radiation gain at the resulting position. Using $d_{\mathrm{BI},m}=\rho_m/\cos\theta_{\mathrm B,m}$ and the LoS path-loss model in~\eqref{eq:pathloss_compact}, the illumination factor $I_m\triangleq\beta_{\mathrm{BI},m}G_{\mathrm B}(\mathbf q_m)$ is
	\begin{equation}
		\begin{aligned}
				\!\!	I_m
			\!\!=\!\!
			C_{\mathrm L}
			\left(\! 
			\frac{\rho_m}{d_0}
			\! \right)^{-\alpha_{\mathrm L}}\!\!
			\left(  
			\cos\theta_{\mathrm B,m}
			  \right)^{\alpha_{\mathrm L}} \!\!
			N_tG_{\mathrm e}(\theta_{\mathrm B,m})
			A_{\mathrm F}(\theta_{\mathrm B,m}).\!\!
		\end{aligned}
		\label{eq:irs_incident_factor_decomp}
	\end{equation}
	For a fixed $\theta_{\mathrm B,m}$, increasing $\rho_m$ enlarges the height-boundary separation linearly while reducing $I_m$ through the BS-to-IRS path-loss factor $\rho_m^{-\alpha_{\mathrm L}}$. Rooftop candidates with a wider main-lobe height band can therefore receive weaker BS illumination when their horizontal distances increase. The main-lobe height range $\mathcal H_m$ determines the retained set $\mathcal M$. For subsequent optimization, the candidate-indexed variables, summations, and constraints in problem~\textrm{(P1)} are restricted from $\mathcal M_0$ to $\mathcal M$ while retaining the complete SNR contributions of the retained candidates. 
	Fig.~\ref{fig:mainlobe_band_validation} compares the exact height-boundary separation in~\eqref{eq:mainlobe_band_width} with the large-array approximation in~\eqref{eq:mainlobe_band_width_scaling} for $N_t\in\{4,8,16\}$. With $\eta_B=3$ dB and $\theta_{\mathrm{tilt}}=8^\circ$, doubling $N_t$ approximately halves the slope of $W_{\eta_B}(\rho_m)$ with respect to $\rho_m$. The relative slope error decreases from $5.60\%$ at $N_t=4$ to $1.37\%$ at $N_t=8$ and $0.34\%$ at $N_t=16$.
	
	\subsection{Directional Characteristics of a Candidate IRS}
	\label{subsec:regional_array_gain}
	For a fixed IRS orientation $\boldsymbol{\tau}_m$, let $\widetilde{\mathbf p}_{m,n}\triangleq\mathbf p_{m,n}-\mathbf q_m$
	denote the position of element $n$ relative to the panel center. The reflected array response associated with a unit outgoing direction $\mathbf e$ is defined as
	\begin{equation}
		\!\!
		A_m(\mathbf e,\boldsymbol{\varphi}_m)
		\!\triangleq\!
		\sum_{n\in\mathcal N_m}
		\exp\Bigg(
		j\varphi_{m,n}
		+
		j\frac{2\pi}{\lambda}
		(\mathbf e_m^{\mathrm{in}}+\mathbf e)^T
		\widetilde{\mathbf p}_{m,n}
		\Bigg).
		\!\!
		\label{eq:array_term_element}
	\end{equation}
	For the outgoing direction $\mathbf e_{m,u}^{\mathrm{out}}$ toward location $u$, the response in~\eqref{eq:array_term_element} reduces to $(\mathbf a_{m,u}^{\mathrm{out}})^T \boldsymbol{\Phi}_m\mathbf a_m^{\mathrm{in}}$
	in~\eqref{eq:cascaded_channel_expand}. Since every element contribution has unit magnitude, the array response satisfies
	$|A_m(\mathbf e,\boldsymbol{\varphi}_m)|^2\leq N_m^2$ for every outgoing direction. A target region generally contains multiple outgoing directions, while one IRS phase vector determines the responses toward all these directions. We therefore let $\mathcal E_m$ denote the set of outgoing directions for the considered portion of the target region and define the corresponding normalized worst-case array gain as $\kappa_m(\boldsymbol{\varphi}_m) \triangleq N_m^{-2}\min_{\mathbf e\in\mathcal E_m} |A_m(\mathbf e,\boldsymbol{\varphi}_m)|^2$. When $\mathcal E_m$ represents the entire target region,
	$\mathcal E_m$ contains the outgoing directions toward all sampled locations and is written as $\mathcal E_m= \{\mathbf e_{m,u}^{\mathrm{out}} \mid\mathbf q_u\in\mathcal Q\}$. Optimizing the IRS phases over this direction set gives the maximum
	achievable normalized worst-case array gain $\kappa_m^{\max}\triangleq \max_{\boldsymbol{\varphi}_m} \kappa_m(\boldsymbol{\varphi}_m)$.
	
	The value of $\kappa_m^{\max}$ depends on the angular separation among the outgoing directions in $\mathcal E_m$ and reflects how widely the target region is viewed from candidate site $m$.
	We measure this separation by twice the radius of the smallest spherical cap containing $\mathcal E_m$ and define
	\begin{equation}
		D_m
		\triangleq
		2\min_{\|\mathbf e_0\|=1}
		\max_{\mathbf e\in\mathcal E_m}
		\angle\left(
		\mathbf e_0,
		\mathbf e
		\right),
		\label{eq:regional_angular_extent}
	\end{equation}
	where $\mathbf e_0$ denotes a candidate unit center direction. The following proposition gives an analytical lower bound on the maximum achievable normalized worst-case array gain over $\mathcal E_m$ in terms of $D_m$ and $N_m$. 
	
	\begin{proposition}
		\label{prop:regional_coherence}
		For a square UPA with $N_{m,h}=N_{m,v}>1$, the maximum achievable normalized worst-case array gain over
		$\mathcal E_m$ satisfies
		\begin{equation}
			\kappa_m^{\max}
			\geq
			\max\left\{
			1-
			\frac{\pi^2(N_m-1)}{6}
			\sin^2\left(\frac{D_m}{4}\right),
			0
			\right\}^2.
			\label{eq:regional_coherence_bound}
		\end{equation}
	\end{proposition}
	
	\begin{proof}
		Let $\mathbf e_m^{\mathrm{ref}}$ denote a minimizing center in~\eqref{eq:regional_angular_extent}, and choose
		$\widetilde{\varphi}_{m,n} =-\frac{2\pi}{\lambda} (\mathbf e_m^{\mathrm{in}}+\mathbf e_m^{\mathrm{ref}})^T
		\widetilde{\mathbf p}_{m,n}$ for every $n\in\mathcal N_m$, with each phase taken modulo $2\pi$. For any $\mathbf e\in\mathcal E_m$, substituting these phases into~\eqref{eq:array_term_element} gives $A_m(\mathbf e,\widetilde{\boldsymbol{\varphi}}_m)  =\sum_{n\in\mathcal N_m}\exp\bigl(j\Delta_{m,n}(\mathbf e)\bigr)$,
		where $\Delta_{m,n}(\mathbf e)  =\frac{2\pi}{\lambda} (\mathbf e-\mathbf e_m^{\mathrm{ref}})^T
		\widetilde{\mathbf p}_{m,n}$. For every $\mathbf e\in\mathcal E_m$,~\eqref{eq:regional_angular_extent} implies
		$\angle(\mathbf e_m^{\mathrm{ref}},\mathbf e)\leq D_m/2$ and hence $\|\mathbf e-\mathbf e_m^{\mathrm{ref}}\|^2
		\leq4\sin^2(D_m/4)$. For the square UPA with $N_{m,h}=N_{m,v}$ and $d_e=\lambda/2$, the centered element coordinates in~\eqref{eq:element_global_position} have equal second moments along the two panel axes and zero cross moment. Their projection onto any vector therefore satisfies $\sum_{n\in\mathcal N_m} [(\mathbf e-\mathbf e_m^{\mathrm{ref}})^T
		\widetilde{\mathbf p}_{m,n}]^2 	\leq \frac{\lambda^2N_m(N_m-1)}{48} \|\mathbf e-\mathbf e_m^{\mathrm{ref}}\|^2$.
		Combining this relation with the preceding direction bound gives
		\begin{equation}
			\sum_{n\in\mathcal N_m}
			\Delta_{m,n}^2(\mathbf e)
			\leq
			\frac{\pi^2N_m(N_m-1)}{3}
			\sin^2\left(\frac{D_m}{4}\right).
			\label{eq:regional_phase_energy_bound}
		\end{equation}
		The response satisfies $\Re\{A_m(\mathbf e,\widetilde{\boldsymbol{\varphi}}_m)\} =\sum_{n\in\mathcal N_m}\cos\Delta_{m,n}(\mathbf e)$. Applying $\cos x\geq1-x^2/2$ and~\eqref{eq:regional_phase_energy_bound}, followed by
		$|A_m|\geq\max\{\Re\{A_m\},0\}$, gives $\kappa_m(\widetilde{\boldsymbol{\varphi}}_m) \geq \max\{1-\frac{\pi^2(N_m-1)}{6}
		\sin^2(D_m/4),0\}^2$. Since $\kappa_m^{\max}$ is the maximum of $\kappa_m(\boldsymbol{\varphi}_m)$ over all feasible IRS phases, $\kappa_m^{\max}\geq \kappa_m(\widetilde{\boldsymbol{\varphi}}_m)$,  which  proves~\eqref{eq:regional_coherence_bound}.
	\end{proof}
	
	\begin{figure}[t!]
		\centering
		\includegraphics[width=3.3in]{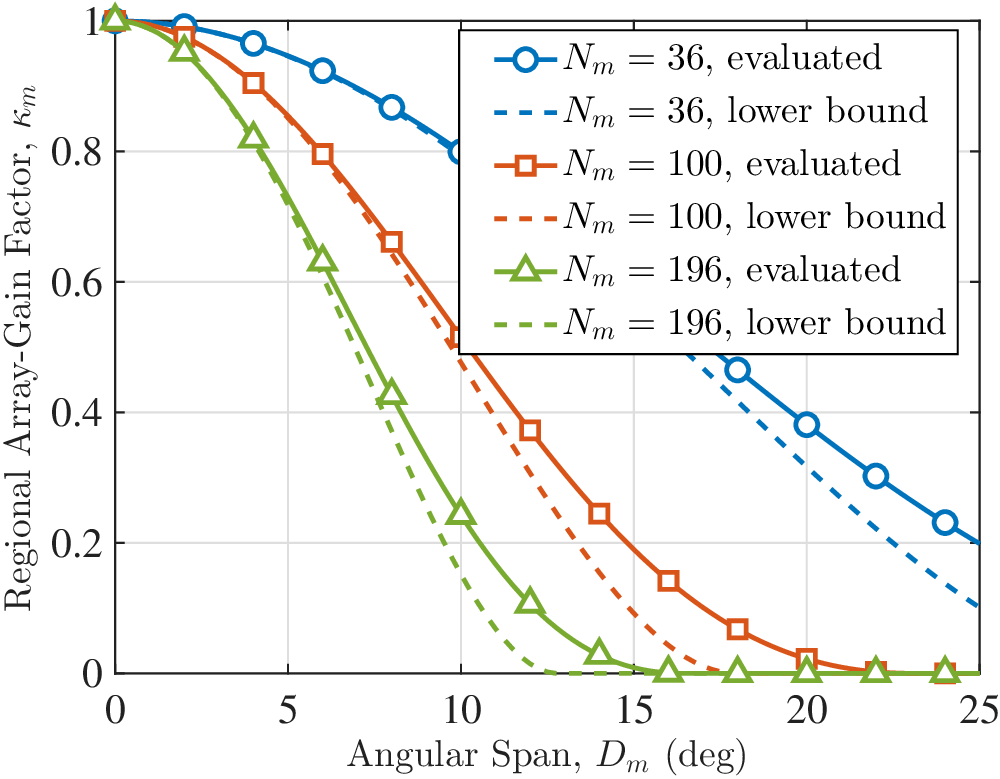}
		\caption{Worst-case normalized array gain of the constructive phase design and analytical lower bound versus $D_m$.}
		\label{fig:coherence}
	\end{figure}
	
	Proposition~\ref{prop:regional_coherence} establishes an analytical lower bound on the maximum achievable normalized worst-case array gain after optimizing one common IRS phase vector over $\mathcal E_m$. The lower bound depends on $N_m$ and $D_m$ through the joint factor $(N_m-1)\sin^2(D_m/4)$ under a fixed IRS orientation. Requiring the lower bound in~\eqref{eq:regional_coherence_bound} to be no smaller than a specified worst-case normalized gain $\bar{\kappa}\in(0,1)$ gives the sufficient direction-span condition
	\begin{equation}
		D_m
		\leq
		4\arcsin
		\sqrt{
			\frac{
				6\left(1-\sqrt{\bar{\kappa}}\right)
			}{
				\pi^2(N_m-1)
			}
		}.
		\label{eq:regional_span_guarantee}
	\end{equation}
	The right-hand side of~\eqref{eq:regional_span_guarantee} decreases monotonically with $N_m$ and gives the maximum direction span allowed by the condition for the specified $\bar{\kappa}$. Using $\arcsin x\sim x$ for small $x$, the sufficient direction span in~\eqref{eq:regional_span_guarantee} scales as $\Theta(N_m^{-1/2})$ for fixed $\bar{\kappa}$. Quadrupling $N_m$ therefore approximately halves the sufficient direction span in the large-array regime under the same worst-case normalized gain requirement. The equivalent condition $(N_m-1)\sin^2(D_m/4)\leq6(1-\sqrt{\bar{\kappa}})/\pi^2$ identifies the array-size and direction-span combinations satisfying the same gain guarantee.
	Fig.~\ref{fig:coherence} compares the worst-case normalized gain of the feasible phase vector used in the proof with the analytical lower bound for the three array sizes. The evaluated curves remain above the lower bounds across the direction spans and confirm the relation in~\eqref{eq:regional_coherence_bound}. At $D_m=12^\circ$ the evaluated gain decreases from $0.7220$ for $N_m=36$ to $0.1064$ for $N_m=196$. The corresponding lower bound decreases from $0.7095$ to $0.0147$ over the array-size range and remains below the evaluated gain. The growing gap with $N_m$ indicates increasing conservatism of the analytical lower bound as the directional phase variation across the array increases.
			
	\section{Alternating Optimization Algorithm}
	\label{sec:alg2}
	
	The optimization is performed over the retained candidate set $\mathcal M$, and each installation height follows~\eqref{eq:height_projection_rule}. For IRS $m$, let $\mathbf v_m\triangleq[e^{j\varphi_{m,1}},\ldots,e^{j\varphi_{m,N_m}}]^T$ denote the unit-modulus reflection vector used in the following alternating updates with $\boldsymbol{\Phi}_m=\operatorname{diag}(\mathbf v_m)$. We propose a mixed-integer AO algorithm for problem~\textrm{(P1)} by decomposing the deployment variables and IRS orientations together with the reflection vectors into three blocks. At outer iteration $k$, $\mathbf s^{(k)}$, $\{\boldsymbol{\tau}_m^{(k)}\}$, and $\{\mathbf v_m^{(k)}\}$ denote the deployment vector, IRS orientations, and reflection vectors, respectively.

	\subsection{Block Updates of the Proposed AO Algorithm}
	\label{subsec:alg2_block_updates}
	\textit{1) Optimizing $\mathbf s$ for Given $\{\boldsymbol{\tau}_m\}$ and $\{\mathbf v_m\}$:}
	With $\{\boldsymbol{\tau}_m^{(k)}\}$ and $\{\mathbf v_m^{(k)}\}$ fixed, the deployment subproblem optimizes $\mathbf s$ and $\Gamma$ subject to~\eqref{prob:maxmin_snr_voxel_region} together with~\eqref{prob:maxmin_raw_budget} and~\eqref{prob:maxmin_raw_binary}. 
	Let $\boldsymbol{\varphi}_m^{(k)}$ denote the phase vector corresponding to $\mathbf v_m^{(k)}$ through $[\mathbf v_m^{(k)}]_n=e^{j\varphi_{m,n}^{(k)}}$ for every $n\in\mathcal N_m$ and let $h_{m,u}^{(k)}\triangleq h_{m,u}(\boldsymbol{\tau}_m^{(k)},\boldsymbol{\varphi}_m^{(k)})$ denote the cascaded channel from candidate site $m$ to location $u$. The identity $s_m^2=s_m$ for each binary deployment variable allows the SNR in~\eqref{eq:snr} at location $\mathbf q_u$ to be written as
	\begin{equation}
		\begin{aligned}
			\gamma_u(\mathbf s)
			=
			\frac{P_0}{\sigma^2}
			\left|
			h_{\mathrm{BV},u}
			+
			\sum\nolimits_{m\in\mathcal M}
			s_m h_{m,u}^{(k)}
			\right|^2 
			&=\frac{P_0}{\sigma^2} \mathcal F_u.
			\label{eq:alg2_coh_expand}
		\end{aligned}
	\end{equation}
	Expanding the squared magnitude in~\eqref{eq:alg2_coh_expand} gives $\mathcal F_u=|h_{\mathrm{BV},u}|^2+\sum_{m\in\mathcal M}a_{m,u}s_m+2\sum_{\substack{m,\ell\in\mathcal M\\m<\ell}}b_{m\ell,u}s_ms_\ell$ for every location $\mathbf q_u\in\mathcal Q$. The coefficient of $s_m$ for each $m\in\mathcal M$ is $a_{m,u}\triangleq|h_{m,u}^{(k)}|^2+2\Re\{h_{\mathrm{BV},u}^*h_{m,u}^{(k)}\}$. The coefficient of $s_ms_\ell$ for every $m,\ell\in\mathcal M$ with $m<\ell$ is $b_{m\ell,u}\triangleq\Re\{(h_{m,u}^{(k)})^*h_{\ell,u}^{(k)}\}$.
	We introduce the auxiliary variable $w_{m\ell}=s_ms_\ell$ for every $m,\ell\in\mathcal M$ with $m<\ell$ to replace each pairwise product in $\mathcal F_u$. Since $s_m$ and $s_\ell$ are binary, each pairwise product for $m,\ell\in\mathcal M$ with $m<\ell$ is represented exactly by
	\begin{equation}
		w_{m\ell}\le s_m,
		w_{m\ell}\le s_\ell,
		w_{m\ell}\ge s_m+s_\ell-1,
		w_{m\ell}\ge 0.
		\label{eq:alg2_product_linearization}
	\end{equation}
	Replacing each product $s_ms_\ell$ in $\mathcal F_u$ with $w_{m\ell}$ gives the linear expression $\widetilde{\mathcal F}_u=|h_{\mathrm{BV},u}|^2+\sum_{m\in\mathcal M}a_{m,u}s_m+2\sum_{\substack{m,\ell\in\mathcal M\\m<\ell}}b_{m\ell,u}w_{m\ell}$ for every $\mathbf q_u\in\mathcal Q$. The auxiliary variables form $\mathbf w\triangleq\{w_{m\ell}\}_{m<\ell}$ and the deployment subproblem under the fixed orientations and reflection coefficients is formulated as
	\begin{subequations}
		\label{prob:alg2_s_sub_full}
		\begin{align}
			\!\! 	\textrm{(P2)} \  
			\max_{\mathbf s,\,\mathbf w,\,\Gamma}\ 
			&\  \Gamma
			\label{prob:alg2_s_sub_full_obj}
			\\
			\mathrm{s.t.}\quad
			&
			\frac{P_0}{\sigma^2}\widetilde{\mathcal F}_u
			\ge
			\Gamma,
			\forall \mathbf q_u\in\mathcal Q,
			\label{prob:alg2_s_sub_full_snr}
			\\
			&
			\eqref{eq:alg2_product_linearization},
			\eqref{prob:maxmin_raw_budget},
			\eqref{prob:maxmin_raw_binary}.
			\label{prob:alg2_s_sub_full_inherit}
		\end{align}
	\end{subequations}
	Problem~\textrm{(P2)} is an MILP with $|\mathcal M|$ binary deployment variables, $\binom{|\mathcal M|}{2}$ auxiliary product variables and $\Gamma$. The binary constraints in~\eqref{prob:maxmin_raw_binary} and the linear constraints in~\eqref{eq:alg2_product_linearization} jointly enforce $w_{m\ell}=s_ms_\ell$ at every feasible point and give $\widetilde{\mathcal F}_u=\mathcal F_u$. Problem~\textrm{(P2)} is therefore an exact MILP reformulation of the deployment subproblem for the fixed IRS orientations and reflection coefficients at outer iteration $k$. Solving problem~\textrm{(P2)} to zero optimality gap by branch-and-bound returns a globally optimal deployment update $\mathbf s^{(k+1)}$ for the current fixed blocks.
	
	\textit{2) Optimizing $\{\boldsymbol{\tau}_m\}$ for Given $\mathbf s$ and $\{\mathbf v_m\}$:}
	With $\mathbf s^{(k+1)}$ and $\{\mathbf v_m^{(k)}\}$ fixed, the orientation block optimizes the selected IRS orientations subject to~\eqref{prob:maxmin_raw_orientation} and the SNR constraints in~\textrm{(P1)}. Let $\mathcal S^{(k+1)}\triangleq\{m\in\mathcal M\mid s_m^{(k+1)}=1\}$ denote the selected IRS set and let $\boldsymbol{\tau}\triangleq\{\boldsymbol{\tau}_m\}_{m\in\mathcal S^{(k+1)}}$ collect the corresponding orientations. Under the fixed deployment and reflection vectors, $\gamma_u(\boldsymbol{\tau})$ denotes the SNR in~\eqref{eq:snr} as a function of the selected IRS orientations and the orientation subproblem at outer iteration $k$ is formulated as
	\begin{subequations}
		\label{eq:alg2_tau_subprob}
		\begin{align}
			\textrm{(P3)}\quad
			\max_{\{\boldsymbol{\tau}_m\}_{m\in\mathcal S^{(k+1)}},\,\Gamma}
			\quad
			& \Gamma
			\label{eq:alg2_tau_subprob_obj}
			\\
			\mathrm{s.t.}\quad
			& \gamma_u(\boldsymbol{\tau})
			\ge
			\Gamma,
			\quad \forall \mathbf q_u\in\mathcal Q,
			\label{eq:alg2_tau_subprob_snr}
			\\
			& \boldsymbol{\tau}_m\in\mathcal T_m,
			\quad \forall m\in\mathcal S^{(k+1)}.
			\label{eq:alg2_tau_subprob_box}
		\end{align}
	\end{subequations}
	The SNR constraints in~\textrm{(P3)} are non-convex in $\{\boldsymbol{\tau}_m\}$ due to the orientation-dependent element gains and array responses in~\eqref{eq:cascaded_channel_expand}. We apply successive convex approximation (SCA) and construct a quadratic local model of each SNR function at $\boldsymbol{\tau}^{(k)}$ to generate an orientation candidate. 
	For each $m\in\mathcal S^{(k+1)}$, define the orientation increment $\mathbf d_m^{\boldsymbol{\tau}}\triangleq\boldsymbol{\tau}_m-\boldsymbol{\tau}_m^{(k)}$ relative to the current outer iterate. Stack these orientation increments as $\mathbf d^{\boldsymbol{\tau}}\triangleq\operatorname{col}\bigl(\{\mathbf d_m^{\boldsymbol{\tau}}\}_{m\in\mathcal S^{(k+1)}}\bigr)$ in increasing order of the candidate index $m$ during the current update. Using these increments, we construct the following concave quadratic local model for the SNR at location $\mathbf q_u$ around $\boldsymbol{\tau}^{(k)}$ as
	\begin{equation}
		\begin{split}
			\!\!\!\!	q_{u}^{\boldsymbol{\tau}}(
			\boldsymbol{\tau}\mid
			\boldsymbol{\tau}^{(k)})
			= & \ 
			\gamma_u(\boldsymbol{\tau}^{(k)})
			\!+\!\!\!\!\!
			\sum_{m\in\mathcal S^{(k+1)}}\!\!\!\!
			\left.
			\nabla_{\boldsymbol{\tau}_m}
			\gamma_u(\boldsymbol{\tau})
			\right|_{\boldsymbol{\tau}=\boldsymbol{\tau}^{(k)}}^T
			\mathbf d_m^{\boldsymbol{\tau}}
			\\
			&\!-\!
			\frac{L_{u}^\tau}{2}
			\sum_{m\in\mathcal S^{(k+1)}}
			\bigl\|
			\mathbf d_m^{\boldsymbol{\tau}}
			\bigr\|^2 ,
		\end{split}
		\label{eq:tau_local_model}
	\end{equation}
	where $L_u^\tau>0$ is a curvature parameter adjusted by backtracking. 
	With $\tau_{m,i}\in\{\theta_m,\phi_m\}$, the gradient in~\eqref{eq:tau_local_model} is evaluated componentwise by differentiating~\eqref{eq:snr} as
	\begin{equation}
		\begin{split}
		\!\!	\left[
			\left.
			\nabla_{\boldsymbol{\tau}_m}
			\gamma_u(\boldsymbol{\tau})
			\right|_{\boldsymbol{\tau}=\boldsymbol{\tau}^{(k)}}
			\right]_i
			&\!=\!
			\left.
			\frac{2P_0}{\sigma^2}
			\Re
			\left\{
			(\bar h_u^{(k)})^*
			\frac{\partial h_{m,u}}
			{\partial\tau_{m,i}}
			\right\}
			\right|_{\boldsymbol{\tau}=\boldsymbol{\tau}^{(k)}}.
		\end{split}\!\!
		\label{eq:tau_gradient_stack}
	\end{equation}
	where $\bar h_u^{(k)}\triangleq h_{\mathrm{BV},u}+\sum_{\ell\in\mathcal S^{(k+1)}}h_{\ell,u}(\boldsymbol{\tau}_\ell^{(k)},\boldsymbol{\varphi}_\ell^{(k)})$ is the total channel before the orientation update under the current deployment and reflection vectors.
	To compute $\partial h_{m,u}/\partial\tau_{m,i}$, write the cascaded channel for candidate site $m$ and location $u$ under the fixed reflection vector $\mathbf v_m^{(k)}$ as
	\begin{equation}
	h_{m,u}
	=
	\zeta_{m,u}
	e^{-j\psi_{m,u}}
	f_{m,u},
	\label{eq:tau_cascaded_decomp}
	\end{equation}
	where
	$f_{m,u}
	=
	\sum_{n\in\mathcal N_m}
	[\mathbf v_m^{(k)}]_n
	[\mathbf a_{m,u}^{\mathrm{out}}]_n
	[\mathbf a_m^{\mathrm{in}}]_n$.
	The propagation phase $\psi_{m,u}$ is independent of $\boldsymbol{\tau}_m$ because all node positions remain fixed during the orientation update. 
	Substituting~\eqref{eq:tau_cascaded_decomp} into~\eqref{eq:tau_gradient_stack} gives~\eqref{eq:tau_grad_from_channel}, shown at the top of the next page. 
	\begin{figure*}
		\begin{equation}
			\begin{split}
				\left.
				\frac{\partial\gamma_u}
				{\partial\tau_{m,i}}
				\right|_{\boldsymbol{\tau}=\boldsymbol{\tau}^{(k)}}
				=
				\left.
				\frac{2P_0}{\sigma^2}
				\Re
				\left\{
				(\bar h_u^{(k)})^*
				e^{-j\psi_{m,u}}
				\left(
				\frac{\partial\zeta_{m,u}}
				{\partial\tau_{m,i}}
				f_{m,u}
				+
				\zeta_{m,u}
				\frac{\partial f_{m,u}}
				{\partial\tau_{m,i}}
				\right)
				\right\}
				\right|_{\boldsymbol{\tau}=\boldsymbol{\tau}^{(k)}} .
			\end{split}
			\label{eq:tau_grad_from_channel}
		\end{equation}
	\end{figure*}
	In~\eqref{eq:tau_grad_from_channel}, $\partial\zeta_{m,u}/\partial\tau_{m,i}$ follows from the ERP gains, while $\partial f_{m,u}/\partial\tau_{m,i}$ follows from the array responses. When the incident and outgoing directions lie in the front half-space of IRS $m$, $\partial\zeta_{m,u}/\partial\tau_{m,i}$ is computed by
	\begin{equation}
		\begin{split}
			\frac{\partial\zeta_{m,u}}
			{\partial\tau_{m,i}}
			=
			\frac{p_I\zeta_{m,u}}{2}
			\bigg(
			&
			\frac{
				\left(
				\frac{\partial\mathbf n_m}
				{\partial\tau_{m,i}}
				\right)^T
				\mathbf e_m^{\mathrm{in}}
			}
			{\cos\theta_m^{\mathrm{in}}}+
			\frac{
				\left(
				\frac{\partial\mathbf n_m}
				{\partial\tau_{m,i}}
				\right)^T
				\mathbf e_{m,u}^{\mathrm{out}}
			}
			{\cos\theta_{m,u}^{\mathrm{out}}}
			\bigg),
		\end{split}
		\label{eq:tau_grad_rho}
	\end{equation}
	where 
	\begin{equation}
		\!\!	\frac{\partial\mathbf n_m}{\partial\theta_m}
		=
		\begin{bmatrix}
			\cos\theta_m\cos\phi_m\\
			\cos\theta_m\sin\phi_m\\
			-\sin\theta_m
		\end{bmatrix}, 
		\frac{\partial\mathbf n_m}{\partial\phi_m}
		=
		\begin{bmatrix}
			-\sin\theta_m\sin\phi_m\\
			\sin\theta_m\cos\phi_m\\
			0
		\end{bmatrix}.\!\!
		\label{eq:alg2_grad_normal}
	\end{equation}
	Because $\zeta_{m,u}\propto(\cos\theta_m^{\mathrm{in}})^{p_I/2}(\cos\theta_{m,u}^{\mathrm{out}})^{p_I/2}$, each term in~\eqref{eq:tau_grad_rho} that contains division by a cosine remains bounded in the front half-space for $p_I\ge2$. If either direction lies in the back half-space, $\zeta_{m,u}$ is locally zero and the derivative with respect to $\tau_{m,i}$ is zero. At the front-side boundary, we assign a zero derivative to $\zeta_{m,u}$ when constructing the quadratic local model for the orientation update. The remaining array-response derivative is given in~\eqref{eq:alg2_grad_f_branch}, shown at the top of the next page,
	\begin{figure*}
		\begin{equation}
			\begin{split}
			\!\! 	\frac{\partial f_{m,u}}
				{\partial\tau_{m,i}}
				\!=\!\!\!
				\sum_{n\in\mathcal N_m}\!
				[\mathbf v_m^{(k)}]_n
				\!\left(\!
				[\mathbf a_m^{\mathrm{in}}]_n
				\frac{\partial[\mathbf a_{m,u}^{\mathrm{out}}]_n}
				{\partial\tau_{m,i}}
				\!+\!
				[\mathbf a_{m,u}^{\mathrm{out}}]_n
				\frac{\partial[\mathbf a_m^{\mathrm{in}}]_n}
				{\partial\tau_{m,i}}
				\!\right)\!
				\!=\!
				j\frac{2\pi}{\lambda}
				\!\!\sum_{n\in\mathcal N_m}\!
				[\mathbf v_m^{(k)}]_n
				[\mathbf a_{m,u}^{\mathrm{out}}]_n
				[\mathbf a_m^{\mathrm{in}}]_n
				\left(
				\mathbf e_m^{\mathrm{in}}
				\!+\!
				\mathbf e_{m,u}^{\mathrm{out}}
				\right)^T
				\frac{\partial\mathbf R_m(\boldsymbol{\tau}_m)}
				{\partial\tau_{m,i}}
				\bar{\mathbf p}_{m,n}.\!
			\end{split}
			\label{eq:alg2_grad_f_branch}
		\end{equation}
	\end{figure*}
	where $ \frac{\partial\mathbf p_{m,n}} {\partial\tau_{m,i}} =
	\frac{\partial\mathbf R_m(\boldsymbol{\tau}_m)}{\partial\tau_{m,i}} \bar{\mathbf p}_{m,n},$ 
	and
	\begin{align}
		\!\!\!\!	\frac{\partial\mathbf R_m}
		{\partial\theta_m}
		&\!=\!
		\begin{bmatrix}
			-\cos\phi_m\sin\theta_m & 0 & \cos\phi_m\cos\theta_m\\
			-\sin\phi_m\sin\theta_m & 0 & \sin\phi_m\cos\theta_m\\
			-\cos\theta_m & 0 & -\sin\theta_m
		\end{bmatrix},
		\label{eq:alg2_grad_R_theta}
		\\
		\!\!\!\!	\frac{\partial\mathbf R_m}
		{\partial\phi_m}
		&\!=\!
		\begin{bmatrix}
			-\sin\phi_m\cos\theta_m \!\!& -\cos\phi_m \!\!& -\sin\phi_m\sin\theta_m\\
			\cos\phi_m\cos\theta_m \!\!& -\sin\phi_m \!\!& \cos\phi_m\sin\theta_m\\
			0\! \!& 0 \!\!& 0
		\end{bmatrix}.\!\!
		\label{eq:alg2_grad_R_phi}
	\end{align}
	Substituting~\eqref{eq:tau_grad_rho} and~\eqref{eq:alg2_grad_f_branch} into~\eqref{eq:tau_grad_from_channel} gives the SNR gradient in~\eqref{eq:tau_local_model} with all quantities evaluated at $\boldsymbol{\tau}=\boldsymbol{\tau}^{(k)}$ during the current orientation update.
	
	Replacing the SNR constraints in~\eqref{eq:alg2_tau_subprob_snr} with the quadratic model in~\eqref{eq:tau_local_model} gives
	\begin{subequations}
		\label{eq:tau_sub}
		\begin{align}
			\textrm{(P3-1)}\ 
			\max_{\mathbf d^{\boldsymbol{\tau}},\Gamma}\quad
			& \Gamma
			\label{eq:tau_sub_obj}
			\\
			\mathrm{s.t.}\ 
			&
			q_{u}^{\boldsymbol{\tau}}
			\left(
			\boldsymbol{\tau}^{(k)}+\mathbf d^{\boldsymbol{\tau}}
			\mid
			\boldsymbol{\tau}^{(k)}
			\right)
			\ge
			\Gamma,
			\forall \mathbf q_u\in\mathcal Q,
			\label{eq:tau_sub_snr}
			\\
			&
			\boldsymbol{\tau}_m^{(k)}
			+
			\mathbf d_m^{\boldsymbol{\tau}}
			\in
			\mathcal T_m,
			\forall m\in\mathcal S^{(k+1)}.
			\label{eq:tau_sub_box}
		\end{align}
	\end{subequations}
	Since $q_u^{\boldsymbol{\tau}}$ is concave quadratic in $\mathbf d^{\boldsymbol{\tau}}$, \textrm{(P3-1)} is a convex quadratically constrained quadratic program (QCQP) under~\eqref{eq:tau_sub_box} and can be solved using CVX. Let $\mathbf d^{\boldsymbol{\tau},\star}$ denote the solution of~\textrm{(P3-1)} and set the orientation candidate as $\boldsymbol{\tau}^{\star}\triangleq\boldsymbol{\tau}^{(k)}+\mathbf d^{\boldsymbol{\tau},\star}$. Because $\mathbf d^{\boldsymbol{\tau}}=\mathbf 0$ is feasible and the local model matches the current SNR, the optimal value of~\textrm{(P3-1)} is no smaller than the current worst-case SNR. 
	The resulting orientation candidate therefore preserves or improves the current worst-case SNR whenever the model value does not exceed the true SNR at every location as expressed by
	\begin{equation}
		q_{u}^{\boldsymbol{\tau}}
		\left(
		\boldsymbol{\tau}^{\star}
		\mid
		\boldsymbol{\tau}^{(k)}
		\right)
		\le
		\gamma_u(\boldsymbol{\tau}^{\star}),
		\quad
		\forall \mathbf q_u\in\mathcal Q.
		\label{eq:tau_check}
	\end{equation}
	 When~\eqref{eq:tau_check} fails at location $u$, increase $L_u^\tau$ as $L_u^\tau\leftarrow\eta_\tau L_u^\tau$ with $\eta_\tau>1$ to lower the local model before solving~\textrm{(P3-1)} again. The candidate is accepted as $\boldsymbol{\tau}^{(k+1)}=\boldsymbol{\tau}^{\star}$ when~\eqref{eq:tau_check} holds at all locations, while $\boldsymbol{\tau}^{(k+1)}=\boldsymbol{\tau}^{(k)}$ is retained if no candidate is accepted within $T_\tau^{\max}$ trials. For every $m\notin\mathcal S^{(k+1)}$, set $\boldsymbol{\tau}_m^{(k+1)}=\boldsymbol{\tau}_m^{\mathrm{br}}\in\mathcal T_m$ with $\boldsymbol{\tau}_m^{\mathrm{br}}$ obtained by clipping the inclination and azimuth associated with the bisector of $\mathbf e_m^{\mathrm{in}}$ and $\mathbf e_m^{\mathrm{ref}}$ before the next deployment update.
	
	\textit{3) Optimizing $\{\mathbf v_m\}$ for Given $\mathbf s$ and $\{\boldsymbol{\tau}_m\}$:}
	With $\mathbf s^{(k+1)}$ and $\{\boldsymbol{\tau}_m^{(k+1)}\}$ fixed, each unit-modulus reflection coefficient is parameterized as $[\mathbf v_m]_n=e^{j\varphi_{m,n}}$ with $\varphi_{m,n}\in\mathbb R$ to represent~\eqref{prob:maxmin_raw_phase} modulo $2\pi$. The total channel at location $u$ becomes
	\begin{equation}
		\bar h_u(\boldsymbol{\varphi})
		=
		h_{\mathrm{BV},u}
		+
		\sum\nolimits_{m\in\mathcal S^{(k+1)}}
		\sum\nolimits_{n\in\mathcal N_m}
		b_{m,u,n}^{(k+1)}
		e^{j\varphi_{m,n}},
		\label{eq:phase_channel_sum}
	\end{equation}
	where $b_{m,u,n}^{(k+1)}$, the coefficient multiplying $e^{j\varphi_{m,n}}$ for element $n$ on IRS $m$, is defined in~\eqref{eq:phase_element_coeff} at the top of the next page,
	\begin{figure*}[t]
		\begin{equation}
			\begin{aligned}
				\!\! 		b_{m,u,n}^{(k+1)}
				& \triangleq 
				\sqrt{
					\beta_{\mathrm{BI},m}
					\beta_{\mathrm{IV},m,u}
					N_t
					G_{\mathrm e}\bigl(\theta_{\mathrm B}(\mathbf q_m)\bigr)
					G_{\mathrm I}(\theta_m^{\mathrm{in}})
					G_{\mathrm I}(\theta_{m,u}^{\mathrm{out}})
				}
				F_{\mathrm B}\bigl(\theta_{\mathrm B}(\mathbf q_m)\bigr)
				e^{-j\frac{2\pi}{\lambda}
					(d_{\mathrm{BI},m}+d_{\mathrm{IV},m,u})}
				[\mathbf a_{m,u}^{\mathrm{out}}]_n
				[\mathbf a_m^{\mathrm{in}}]_n .\!\!
			\end{aligned}
			\label{eq:phase_element_coeff}
		\end{equation}
		\hrulefill
	\end{figure*}
	The elements of the selected IRSs are stacked in increasing $m$ and $n$ with the total number $N_{\mathcal S}\triangleq\sum_{m\in\mathcal S^{(k+1)}}N_m$. For stacked index $i$ corresponding to element $n_i$ on IRS $m_i$, define $b_{u,i}^{(k+1)}\triangleq b_{m_i,u,n_i}^{(k+1)}$ and $\varphi_i\triangleq\varphi_{m_i,n_i}$ with $\boldsymbol{\varphi}\triangleq[\varphi_1,\ldots,\varphi_{N_{\mathcal S}}]^T$.  Equation~\eqref{eq:phase_channel_sum} then becomes
	\begin{equation}
		\bar h_u(\boldsymbol{\varphi})
		=
		h_{\mathrm{BV},u}
		+
		\sum\nolimits_{i=1}^{N_{\mathcal S}}
		b_{u,i}^{(k+1)}
		e^{j\varphi_i}.
		\label{eq:f_phase}
	\end{equation}
	The corresponding SNR is
	\begin{equation}
		\gamma_u(\boldsymbol{\varphi}) = \frac{P_0}{\sigma^2} \left| \bar h_u(\boldsymbol{\varphi}) \right|^2 .
		\label{eq:alg2_gamma_v_quadratic}
	\end{equation}
	Substituting~\eqref{eq:alg2_gamma_v_quadratic} into the SNR constraint of~\textrm{(P1)} with $\mathbf s^{(k+1)}$ and $\{\boldsymbol{\tau}_m^{(k+1)}\}$ fixed gives the phase-shift subproblem
	\begin{subequations}
		\label{eq:alg2_v_subprob}
		\begin{align}
			\textrm{(P4)}\quad
			\max_{\boldsymbol{\varphi},\,\Gamma}
			\quad
			& \Gamma
			\label{eq:alg2_v_subprob_obj}
			\\
			\mathrm{s.t.}\quad
			& \gamma_u(\boldsymbol{\varphi})\ge\Gamma,
			\quad \forall \mathbf q_u\in\mathcal Q.
			\label{eq:alg2_v_subprob_snr}
		\end{align}
	\end{subequations}
	The trigonometric polynomial $\gamma_u(\boldsymbol{\varphi})$ is generally non-concave. With $\mathbf d^{\boldsymbol{\varphi}}\triangleq\boldsymbol{\varphi}-\boldsymbol{\varphi}^{(k)}$, the quadratic local model of $\gamma_u(\boldsymbol{\varphi})$ at $\boldsymbol{\varphi}^{(k)}$ is
	\begin{equation}
		\begin{aligned}
			q_u^{\boldsymbol{\varphi}}
			(
			\boldsymbol{\varphi}
			\mid
			\boldsymbol{\varphi}^{(k)}
			)
			&=
			\gamma_u(\boldsymbol{\varphi}^{(k)})
			+
			(\mathbf g_u^{\boldsymbol{\varphi}})^T
			\mathbf d^{\boldsymbol{\varphi}}
			-
			\frac{L_u^\varphi}{2}
			\|\mathbf d^{\boldsymbol{\varphi}}\|^2,
		\end{aligned}
		\label{eq:phi_local_model}
	\end{equation}
	where $\mathbf g_u^{\boldsymbol{\varphi}}\triangleq\left.\nabla_{\boldsymbol{\varphi}}\gamma_u(\boldsymbol{\varphi})\right|_{\boldsymbol{\varphi}=\boldsymbol{\varphi}^{(k)}}$ and $L_u^\varphi>0$ is a curvature parameter adjusted by backtracking. Equation~\eqref{eq:f_phase} gives $\partial\bar h_u(\boldsymbol{\varphi})/\partial\varphi_i=jb_{u,i}^{(k+1)}e^{j\varphi_i}$ and $\partial\gamma_u(\boldsymbol{\varphi})/\partial\varphi_i=-(2P_0/\sigma^2)\Im\{\bar h_u^*(\boldsymbol{\varphi})b_{u,i}^{(k+1)}e^{j\varphi_i}\}$. Evaluating the SNR derivative at $\boldsymbol{\varphi}^{(k)}$ gives the $i$-th entry of $\mathbf g_u^{\boldsymbol{\varphi}}$ as
	\begin{equation}
		\!\!	\left[
		\mathbf g_u^{\boldsymbol{\varphi}}
		\right]_i
		\!	=\!
		-\!
		\frac{2P_0}{\sigma^2}
		\Im
		\left\{
		(\bar h_u^{(k)})^*
		b_{u,i}^{(k+1)}
		e^{j\varphi_i^{(k)}}
		\right\}, 
		i=1,\ldots,N_{\mathcal S},
		\label{eq:phase_gradient_entry}
	\end{equation}
	where $\bar h_u^{(k)}\triangleq\bar h_u(\boldsymbol{\varphi}^{(k)})$ and $\Im\{\cdot\}$ denotes the imaginary part.
	Replacing~\eqref{eq:alg2_v_subprob_snr} in~\textrm{(P4)} with the quadratic model in~\eqref{eq:phi_local_model} using~\eqref{eq:phase_gradient_entry} gives
	\begin{subequations}
		\label{eq:phi_sub}
		\begin{align}
			\textrm{(P4-1)}\quad
			\max_{\boldsymbol{\varphi},\,\Gamma}
			\quad
			& \Gamma
			\label{eq:phi_sub_obj}
			\\
			\mathrm{s.t.}\quad
			&
			q_u^{\boldsymbol{\varphi}}
			(
			\boldsymbol{\varphi}
			\mid
			\boldsymbol{\varphi}^{(k)}
			)
			\ge
			\Gamma,
			\quad \forall \mathbf q_u\in\mathcal Q .
			\label{eq:phi_sub_snr}
		\end{align}
	\end{subequations}
	Problem~\textrm{(P4-1)} is a convex QCQP because $q_u^{\boldsymbol{\varphi}}$ is concave and the max-min structure enables the lower-complexity solver derived below.
	
	Problem~\textrm{(P4-1)} maximizes the minimum quadratic-model value over all sampled locations. The phase increment is therefore obtained from
	\begin{equation}
		\max_{\mathbf d^{\boldsymbol{\varphi}}}
		\min_{\mathbf q_u\in\mathcal Q}
		q_u^{\boldsymbol{\varphi}}
		(
		\boldsymbol{\varphi}^{(k)}
		+
		\mathbf d^{\boldsymbol{\varphi}}
		\mid
		\boldsymbol{\varphi}^{(k)}
		).
		\label{eq:phi_max_min_d}
	\end{equation}
	For the worst-case term in~\eqref{eq:phi_max_min_d}, introduce the simplex variable $\mathbf p=[p_u]_{\mathbf q_u\in\mathcal Q}$, where $p_u$ weights location $\mathbf q_u$, and
	\begin{equation}
		\Delta
		\triangleq
		\left\{
		\mathbf p\in\mathbb R^{|\mathcal Q|}
		\mid
		p_u\ge0,
		\sum\nolimits_{\mathbf q_u\in\mathcal Q}p_u=1
		\right\}.
		\label{eq:phase_simplex}
	\end{equation}
	Problem~\textrm{(P4-1)} maximizes the linear objective $\Gamma$ subject to $\Gamma\le q_u^{\boldsymbol{\varphi}}(\boldsymbol{\varphi}^{(k)}+\mathbf d^{\boldsymbol{\varphi}}\mid\boldsymbol{\varphi}^{(k)})$ for all $\mathbf q_u\in\mathcal Q$. Since $\mathbf d^{\boldsymbol{\varphi}}=\mathbf0$ with any $\Gamma<\min_{\mathbf q_u\in\mathcal Q}\gamma_u(\boldsymbol{\varphi}^{(k)})$ is strictly feasible, Slater's condition holds. 
	Introducing the multiplier $p_u\ge0$ for each constraint gives the Lagrangian
	\begin{equation}
		\mathcal L
	=
	\Gamma
	+
	\sum_{\mathbf q_u\in\mathcal Q}p_u
	\left(
	q_u^{\boldsymbol{\varphi}}
	(\boldsymbol{\varphi}^{(k)}+\mathbf d^{\boldsymbol{\varphi}}
	\mid\boldsymbol{\varphi}^{(k)})
	-\Gamma
	\right).
	\end{equation}
	The resulting dual function is finite only when the multipliers satisfy the simplex constraint $\mathbf p\in\Delta$ in~\eqref{eq:phase_simplex}. Strong duality then gives
	\begin{equation}
		\begin{aligned}
			\max_{\mathbf d^{\boldsymbol{\varphi}}}
			\min_{\mathbf q_u\in\mathcal Q}
			&	q_u^{\boldsymbol{\varphi}}
			(
			\boldsymbol{\varphi}^{(k)}
			+
			\mathbf d^{\boldsymbol{\varphi}}
			\mid
			\boldsymbol{\varphi}^{(k)}
			)
			\\
			&=
			\min_{\mathbf p\in\Delta}
			\max_{\mathbf d^{\boldsymbol{\varphi}}}
			\sum_{\mathbf q_u\in\mathcal Q}
			p_u
			q_u^{\boldsymbol{\varphi}}
			(
			\boldsymbol{\varphi}^{(k)}
			+
			\mathbf d^{\boldsymbol{\varphi}}
			\mid
			\boldsymbol{\varphi}^{(k)}
			).
		\end{aligned}
		\label{eq:phi_dual_minmax}
	\end{equation}
	For a given $\mathbf p\in\Delta$, the objective of the maximization over $\mathbf d^{\boldsymbol{\varphi}}$ in~\eqref{eq:phi_dual_minmax} is
	\begin{equation}
		\begin{aligned}
			&
			\sum\nolimits_{\mathbf q_u\in\mathcal Q}
			p_u
			q_u^{\boldsymbol{\varphi}}
			(
			\boldsymbol{\varphi}^{(k)}
			+
			\mathbf d^{\boldsymbol{\varphi}}
			\mid
			\boldsymbol{\varphi}^{(k)}
			)
			\\
			&=
			c(\mathbf p)
			+
			(\mathbf g^{\boldsymbol{\varphi}}(\mathbf p))^T
			\mathbf d^{\boldsymbol{\varphi}}
			-
			\frac{W(\mathbf p)}{2}
			\|
			\mathbf d^{\boldsymbol{\varphi}}
			\|^2,
		\end{aligned}
		\label{eq:phase_weighted_model}
	\end{equation}
	where $\mathbf g^{\boldsymbol{\varphi}}(\mathbf p)\triangleq\sum_{\mathbf q_u\in\mathcal Q}p_u\mathbf g_u^{\boldsymbol{\varphi}}$, $W(\mathbf p)\triangleq\sum_{\mathbf q_u\in\mathcal Q}p_u L_u^\varphi$, and
	$c(\mathbf p)\triangleq\sum_{\mathbf q_u\in\mathcal Q}p_u\gamma_u(\boldsymbol{\varphi}^{(k)})$.
	The objective in~\eqref{eq:phase_weighted_model} is concave quadratic because $W(\mathbf p)>0$ and has gradient $\mathbf g^{\boldsymbol{\varphi}}(\mathbf p)-W(\mathbf p)\mathbf d^{\boldsymbol{\varphi}}$. Setting the gradient to zero gives the maximizer in~\eqref{eq:phi_dual_minmax} as
	\begin{equation}
		\mathbf d^{\boldsymbol{\varphi},\star}(\mathbf p)
		=\mathbf g^{\boldsymbol{\varphi}}(\mathbf p)
		(W(\mathbf p))^{-1}.
		\label{eq:phi_dual_dstar}
	\end{equation}
	Substituting~\eqref{eq:phi_dual_dstar} into~\eqref{eq:phase_weighted_model} and the outer minimization in~\eqref{eq:phi_dual_minmax} reduces to
	\begin{equation}
		\min_{\mathbf p\in\Delta} V(\mathbf p) \triangleq \min_{\mathbf p\in\Delta} c(\mathbf p)+\|\mathbf g^{\boldsymbol{\varphi}}(\mathbf p)\|^2/(2W(\mathbf p)).
		\label{eq:phase_dual_problem}
	\end{equation}
	The affine dependence of $c(\mathbf p)$, $\mathbf g^{\boldsymbol{\varphi}}(\mathbf p)$, and $W(\mathbf p)$ on $\mathbf p$ with $W(\mathbf p)>0$ makes $V(\mathbf p)$ convex over $\Delta$. Since the inner maximizer $\mathbf d^{\boldsymbol{\varphi},\star}(\mathbf p)$ in~\eqref{eq:phi_dual_dstar} is unique, the gradient required to minimize $V(\mathbf p)$ is given by the Danskin theorem~\cite{danskin1967maxmin} as
	\begin{equation}
		\left[
		\nabla V(\mathbf p)
		\right]_u
		=
		q_u^{\boldsymbol{\varphi}}
		(
		\boldsymbol{\varphi}^{(k)}
		+
		\mathbf d^{\boldsymbol{\varphi},\star}(\mathbf p)
		\mid
		\boldsymbol{\varphi}^{(k)}
		),
		\quad \mathbf q_u\in\mathcal Q .
		\label{eq:phase_dual_grad}
	\end{equation}
	
	Problem~\eqref{eq:phase_dual_problem} is therefore solved by projected gradient descent (PGD) with iterate
	\begin{equation}
		\mathbf p^{(t+1)}
		=
		\Pi_{\Delta}
		\left(
		\mathbf p^{(t)}
		-
		\chi_t
		\nabla V(\mathbf p^{(t)})
		\right),
		\label{eq:phi_dual_pg}
	\end{equation}
	where $\Pi_\Delta(\cdot)$ denotes the Euclidean projection onto $\Delta$ and the step size $\chi_t>0$ is set by backtracking. The iteration stops when the PGD residual falls below $\epsilon_{\mathrm{pgd}}$ or the iteration count reaches $T_{\mathrm{pgd}}$ and returns $\mathbf p^\star$. The candidate phase increment is $\mathbf d^{\boldsymbol{\varphi},\star}=\mathbf d^{\boldsymbol{\varphi},\star}(\mathbf p^\star)$ and the corresponding phase vector is $\boldsymbol{\varphi}^{\star}=\boldsymbol{\varphi}^{(k)}+\mathbf d^{\boldsymbol{\varphi},\star}$.
	The candidate is accepted only when
	\begin{equation}
		q_u^{\boldsymbol{\varphi}}
		(
		\boldsymbol{\varphi}^\star
		\mid
		\boldsymbol{\varphi}^{(k)}
		)
		\le
		\gamma_u(\boldsymbol{\varphi}^\star),
		\quad
		\forall \mathbf q_u\in\mathcal Q,
		\label{eq:phi_check}
	\end{equation}
	and
	\begin{equation}
		\min_{\mathbf q_u\in\mathcal Q}
		\gamma_u(\boldsymbol{\varphi}^\star)
		\ge
		\min_{\mathbf q_u\in\mathcal Q}
		\gamma_u(\boldsymbol{\varphi}^{(k)}).
		\label{eq:phi_true_accept}
	\end{equation}
	When~\eqref{eq:phi_check} fails at location $u$, the model overestimates the coherent SNR at $\boldsymbol{\varphi}^\star$, and the curvature parameter is enlarged as $L_u^\varphi\leftarrow\eta_\varphi L_u^\varphi$ with $\eta_\varphi>1$. Problem~\textrm{(P4-1)} is then solved again, and this enlargement is repeated for at most $T_\varphi^{\max}$ iterations if the check remains unsatisfied. If~\eqref{eq:phi_check} holds at all locations and~\eqref{eq:phi_true_accept} is satisfied, the phase update is accepted as $\boldsymbol{\varphi}^{(k+1)}=\mathrm{mod}(\boldsymbol{\varphi}^{\star},2\pi)$. Otherwise, the previous phase is retained as $\boldsymbol{\varphi}^{(k+1)}=\boldsymbol{\varphi}^{(k)}$. For $m\in\mathcal S^{(k+1)}$, the reflection coefficient of element $n$ is recovered as $[\mathbf v_m^{(k+1)}]_n=e^{j\varphi_i^{(k+1)}}$, where $i$ is the corresponding stacked index. For $m\notin\mathcal S^{(k+1)}$, the reflection coefficients are aligned at the reference direction through $\varphi_{m,n}=-\frac{2\pi}{\lambda}(\mathbf e_m^{\mathrm{in}}+\mathbf e_m^{\mathrm{ref}})^T\widetilde{\mathbf p}_{m,n}$. 

	\subsection{Overall Procedure, Convergence, and Complexity}
	\label{subsec:alg2_summary}

	\begin{algorithm}[t]
		\caption{Mixed-Integer AO Algorithm}
		\label{alg:alg2}
		\footnotesize
		\begin{algorithmic}[1]
			\STATE \textbf{Input} $\mathcal M$, $\mathcal Q$, $M_{\max}$, $\{\mathcal T_m\}$, feasible $\mathbf s^{(0)}$, $\{\boldsymbol{\tau}_m^{(0)},\boldsymbol{\varphi}_m^{(0)}\}_{m\in\mathcal M}$, tolerance $\epsilon_{\mathrm{in}}$ (in dB), and $K_{\max}$.
			\STATE \textbf{Parameters} $\{L_u^\tau,L_u^\varphi\}_{\mathbf q_u\in\mathcal Q}$, $\eta_\tau$, $\eta_\varphi$, $T_\tau^{\max}$, $T_\varphi^{\max}$, $T_{\mathrm{pgd}}$, and $\epsilon_{\mathrm{pgd}}$.
			\STATE Set $[\mathbf v_m^{(0)}]_n=e^{j\varphi_{m,n}^{(0)}}$ for all $m\in\mathcal M$ and $n\in\mathcal N_m$, set $k\gets0$, and compute $\gamma_{\min}^{(0)}$.
			\REPEAT
			\STATE Solve~\eqref{prob:alg2_s_sub_full} for fixed $\{\boldsymbol{\tau}_m^{(k)}\}$ and $\{\mathbf v_m^{(k)}\}$ to obtain $\mathbf s^{(k+1)}$ and $\mathcal S^{(k+1)}$.
			\STATE Update $\boldsymbol{\tau}^{(k+1)}$ by solving~\eqref{eq:tau_sub} with the backtracking condition in~\eqref{eq:tau_check}. Retain $\boldsymbol{\tau}^{(k)}$ when no candidate satisfies~\eqref{eq:tau_check}.
			\STATE Update $\boldsymbol{\varphi}^{(k+1)}$ by solving~\eqref{eq:phase_dual_problem} with the checks in~\eqref{eq:phi_check} and~\eqref{eq:phi_true_accept}. Retain $\boldsymbol{\varphi}^{(k)}$ when no candidate is accepted.
			\STATE Recover $\{\mathbf v_m^{(k+1)}\}_{m\in\mathcal S^{(k+1)}}$ from $\boldsymbol{\varphi}^{(k+1)}$ and assign the reference-direction phases to $m\notin\mathcal S^{(k+1)}$.
			\STATE Compute $\gamma_{\min}^{(k+1)}$ and set $k\gets k+1$.
			\UNTIL{$|10\log_{10}\gamma_{\min}^{(k)}-10\log_{10}\gamma_{\min}^{(k-1)}|\le\epsilon_{\mathrm{in}}$ or $k=K_{\max}$}
			\STATE \textbf{Output} $\mathbf s^\star=\mathbf s^{(k)}$, $\{\boldsymbol{\tau}_m^\star\}=\{\boldsymbol{\tau}_m^{(k)}\}$, $\{\mathbf v_m^\star\}=\{\mathbf v_m^{(k)}\}$, and $\gamma_{\min}^\star=\gamma_{\min}^{(k)}$.
		\end{algorithmic}
	\end{algorithm}

	The overall procedure is summarized in Algorithm~\ref{alg:alg2}. 
	Each block update preserves or increases $\gamma_{\min}(\mathbf s,\boldsymbol{\tau},\boldsymbol{\varphi})$. Deployment $\mathbf s^{(k)}$ is feasible for~\eqref{prob:alg2_s_sub_full}, yielding a non-decreasing deployment update. The orientation update is non-decreasing because $\mathbf d^{\boldsymbol{\tau}}=\mathbf 0$ is feasible for~\textrm{(P3-1)}, the local model equals $\gamma_u$ at $\boldsymbol{\tau}^{(k)}$, and every accepted candidate satisfies~\eqref{eq:tau_check}. A phase candidate is accepted only when~\eqref{eq:phi_true_accept} holds, while $\boldsymbol{\varphi}^{(k)}$ is retained otherwise. The merit sequence ${\gamma_{\min}^{(k)}}$ is therefore monotonically non-decreasing. Each $\gamma_u$ is a finite sum of bounded terms at finite transmit power, and $\min_{\mathbf q_u\in\mathcal Q}\gamma_u$ is bounded above. The merit sequence converges.
	
	Per-iteration complexity combines the three block costs. The deployment MILP in~\eqref{prob:alg2_s_sub_full} has $M$ binary variables, $\binom{M}{2}$ product variables, and $|\mathcal Q|$ constraints, with branch-and-bound cost exponential in $M$. The orientation block evaluates the SNR gradients in~\eqref{eq:tau_grad_from_channel} at $\mathcal O(|\mathcal S||\mathcal Q|\bar N)$ operations, with $\bar N$ the average number of elements per selected IRS, and solves~\eqref{eq:tau_sub} at most $T_\tau^{\max}$ times. The phase block computes per-element coefficients and gradients at $\mathcal O(N_{\mathcal S}|\mathcal Q|)$ operations, with $N_{\mathcal S}=\sum_{m\in\mathcal S}N_m$, and runs $T_{\mathrm{pgd}}$ PGD steps for~\eqref{eq:phase_dual_problem}. The orientation and phase blocks scale polynomially, while the deployment MILP scales exponentially in $M$ and dominates the per-iteration cost for large $M$. Combining the three block costs, the per-iteration complexity is $\mathcal O\bigl(2^{M}|\mathcal Q|+T_\tau^{\max}|\mathcal S||\mathcal Q|\bar N+T_{\mathrm{pgd}}N_{\mathcal S}|\mathcal Q|\bigr)$. With at most $K_{\max}$ outer iterations, the overall complexity is $\mathcal O\bigl(K_{\max}\bigl[2^{M}|\mathcal Q|+T_\tau^{\max}|\mathcal S||\mathcal Q|\bar N+T_{\mathrm{pgd}}N_{\mathcal S}|\mathcal Q|\bigr]\bigr)$.

	\section{Numerical Results}
	\label{sec:results}
	
	\subsection{Simulation Setup and Baselines}
	\label{sec:setup}
	We consider a $300$~m $\times$ $300$~m urban area with a $4\times4$ grid of buildings whose horizontal centers lie at $(x,y)\in\{-90,-30,30,90\}~\text{m}\times\{-90,-30,30,90\}~\text{m}$, with a footprint of $40$~m$\times40$~m each. The initial candidate set $\mathcal M_0$ contains one mast-mounted rooftop-center IRS site per building. The target airspace covers $[-70,70]$~m$\times[-70,70]$~m and the altitude interval $[50,110]$~m.
	Unless otherwise specified, the system parameters are set as follows. $f_c=3.5$~GHz, $B=20$~MHz, $H_B=35$~m, $N_t=8$, $d_v=\lambda/2$, $\theta_{\mathrm{tilt}}=8^\circ$, $P_0=37$~dBm, $\sigma^2=-92$~dBm, $G_{\mathrm e}^{\max}=8$~dBi, $\theta_{\mathrm{3dB}}=65^\circ$, $\mathrm{SLA}_V=30$~dB, $\alpha_{\mathrm L}=2.2$, $d_0=1$~m, $\Delta_v=10$~m. The rooftop heights $H_{b(m)}$ are drawn independently and uniformly from $[18,25]$~m. With $\eta_B=3$~dB, the analytical main-lobe screening determines the retained candidate set $\mathcal M$, and the installation height at each retained candidate follows the projection in~\eqref{eq:height_projection_rule}. The sampling interval gives $14\times14\times6=1176$ basic locations, which are augmented with the predicted direct-link null heights to form $\mathcal Q$.
	The main coverage experiments use a square UPA with $N_m=100$ reflecting elements and half-wavelength spacing $d_e=\lambda/2$ unless the array size is explicitly varied. The element ERP in~\eqref{eq:irs_erp} uses $p_I=2$ with $G_{\mathrm I}^{\max}=2(p_I+1)$. The validation of Proposition~\ref{prop:regional_coherence} considers $N_m\in\{36,100,196\}$.  The orientation bounds in~\eqref{eq:orientation_feasible_set} are $\theta_m\in[0^\circ,90^\circ]$ and $\phi_m\in[0^\circ,360^\circ)$ for all candidates, while $M_{\max}\in\{1,2,3,4,6,8,10\}$.
	We compare the proposed scheme with six coverage-design baselines listed as follows:

	\begin{itemize}
		\item \textit{No-IRS}. Only the BS direct link is used without any IRS.
		
		\item \textit{Random-Site}. Ten trials randomly select $M_{\max}$ sites from $\mathcal M$ and optimize the corresponding IRS orientations and phase shifts as in the proposed scheme. The worst-case SNR is averaged in the linear scale and reported in decibels.
			
		\item \textit{Max-BI-Power}. This scheme selects the $M_{\max}$ rooftop candidates with the largest BS-to-IRS incident power $\beta_{\mathrm{BI},m}G_{\mathrm B}(\mathbf q_m)$ and optimizes the orientations and phase shifts of the selected IRSs.
		
		\item \textit{Centroid Phase}. The deployment and orientations follow the proposed scheme, and the cascaded coefficients are co-phased toward the reflection-strength-weighted centroid of the associated front-half-space locations.
		
		\item \textit{ERP-Agnostic}. The orientations and phase shifts are optimized under $G_{\mathrm I}\equiv1$, and the resulting coverage is evaluated using the IRS element ERP in~\eqref{eq:irs_erp}.
		
		\item \textit{Fixed-Tilt}. All IRS orientations remain fixed at the initial values $\boldsymbol{\tau}_m^{(0)}$, and the deployment indicators and phase shifts are optimized.
	\end{itemize}
	
	\subsection{Validation of Analytical Results and Algorithm}
	
		\begin{figure}[t!]
		\centering
		\includegraphics[width=3.4   in]{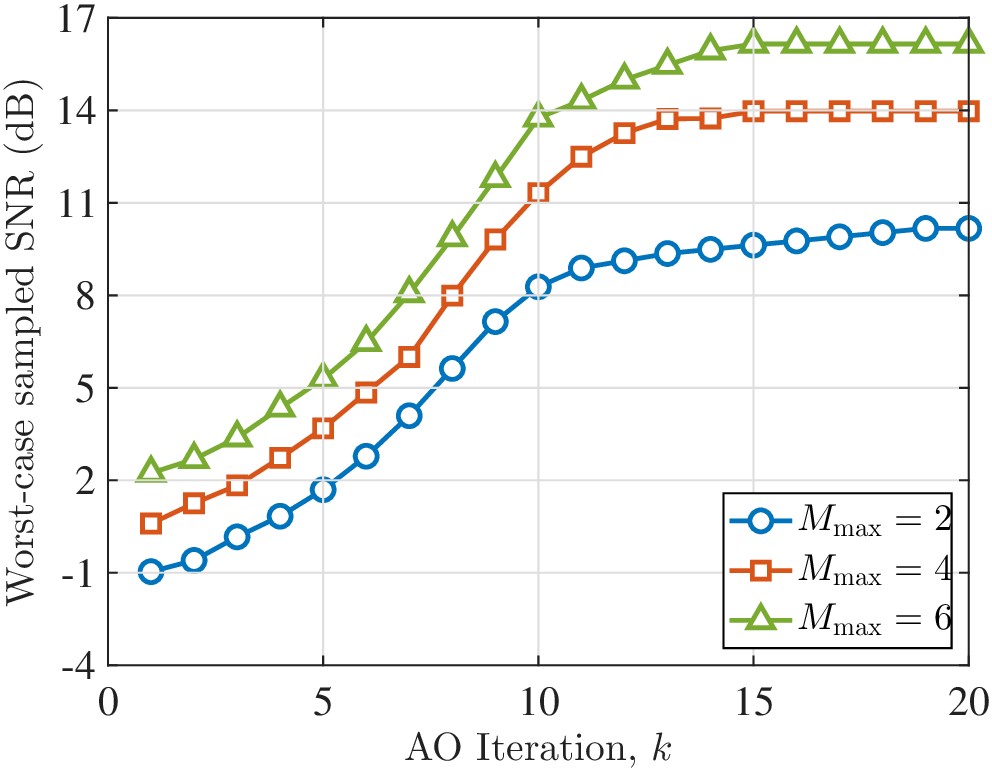}
		\caption{Convergence of the AO algorithm.}
		\label{fig:convergence}
	\end{figure}
	
	Fig.~\ref{fig:convergence} shows the convergence of the proposed AO algorithm for $M_{\max}=2$, $4$, and $6$. The worst-case SNR increases monotonically without oscillation for all deployment budgets. The worst-case SNR rises rapidly during the first several iterations as the deployment, IRS orientations, and phase shifts are jointly improved, and the increments gradually decrease as the three variable blocks become adjusted. The worst-case SNR for all three deployment budgets approaches a stable value within about ten to twelve outer iterations, indicating that the proposed algorithm converges within a moderate number of iterations. The converged SNR increases with $M_{\max}$, and the case with $M_{\max}=6$ remains above those with $M_{\max}=4$ and $2$ near convergence because a larger deployment budget provides more reflected links for improving spatially separated weak-coverage locations. 
	
	\begin{figure}[t!]
		\centering
		\includegraphics[width=3.5  in]{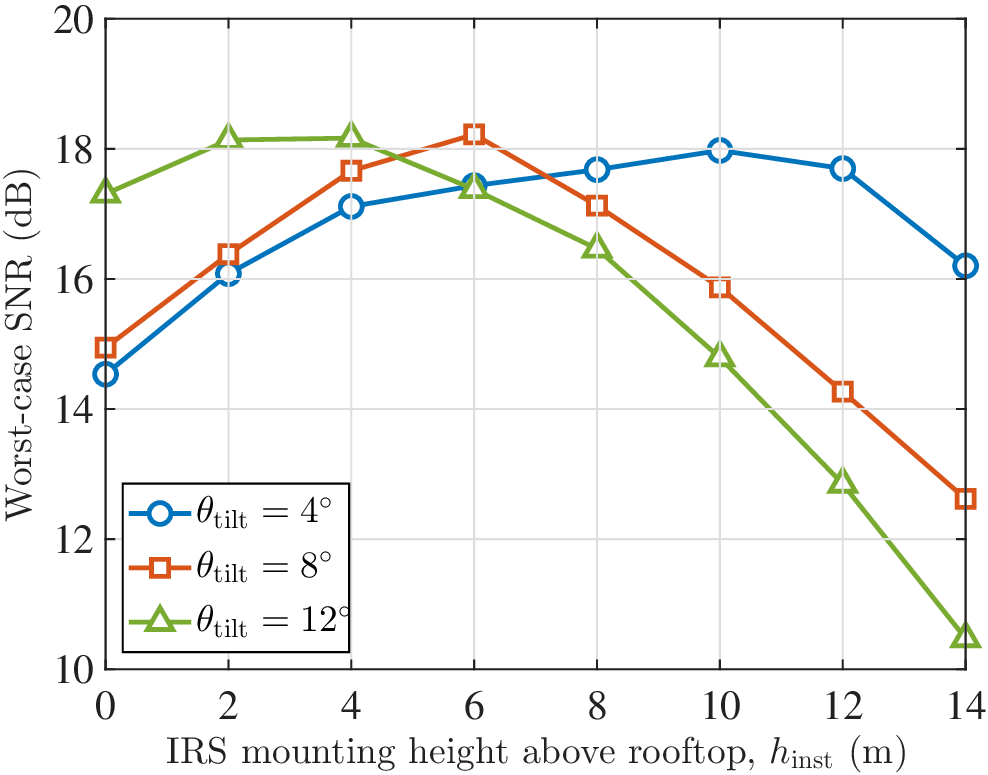}
		\caption{Worst-case SNR versus the IRS mast height $h_{\mathrm{inst}}$.}
		\label{fig:height}
	\end{figure}
	In Fig.~\ref{fig:height}, we plot the worst-case SNR versus the IRS mast height above the rooftop $h_{\mathrm{inst}}\in[0,14]$~m for $\theta_{\mathrm{tilt}}=4^\circ$, $8^\circ$, and $12^\circ$ with $M_{\max}=6$. Every candidate site uses the same mast height during this sweep, whereas the default configuration assigns a separate installation height to each retained site through the projection rule. The IRS-center height is $q_{m,z}=H_{b(m)}+h_{\mathrm{inst}}$, and the worst-case SNR first increases and then decreases with $h_{\mathrm{inst}}$ for all three downtilts. This variation mainly follows the change in BS array gain as the IRS centers approach and then move away from the main-lobe height $z=H_B-\rho\tan\theta_{\mathrm{tilt}}$. The height variation also changes the IRS-to-location distances and outgoing element gains. The favorable height range shifts downward as the downtilt increases, and the peak moves from $h_{\mathrm{inst}}=10$~m at $\theta_{\mathrm{tilt}}=4^\circ$ to $h_{\mathrm{inst}}=4$~m at $\theta_{\mathrm{tilt}}=12^\circ$. At the inner candidate distance $\rho_m=42.4$~m, the downtilt-aligned panel-center height $z_{\mathrm c}(\rho_m)$ equals $32.0$~m and $26.0$~m for these two downtilts. The predicted separation of $6.05$~m between these heights is consistent with the $6$~m shift between the corresponding SNR peaks.
	
  	\subsection{Coverage Performance and Design Insights}
	
	\begin{figure}[t!]
		\centering
		\includegraphics[width=3.4  in]{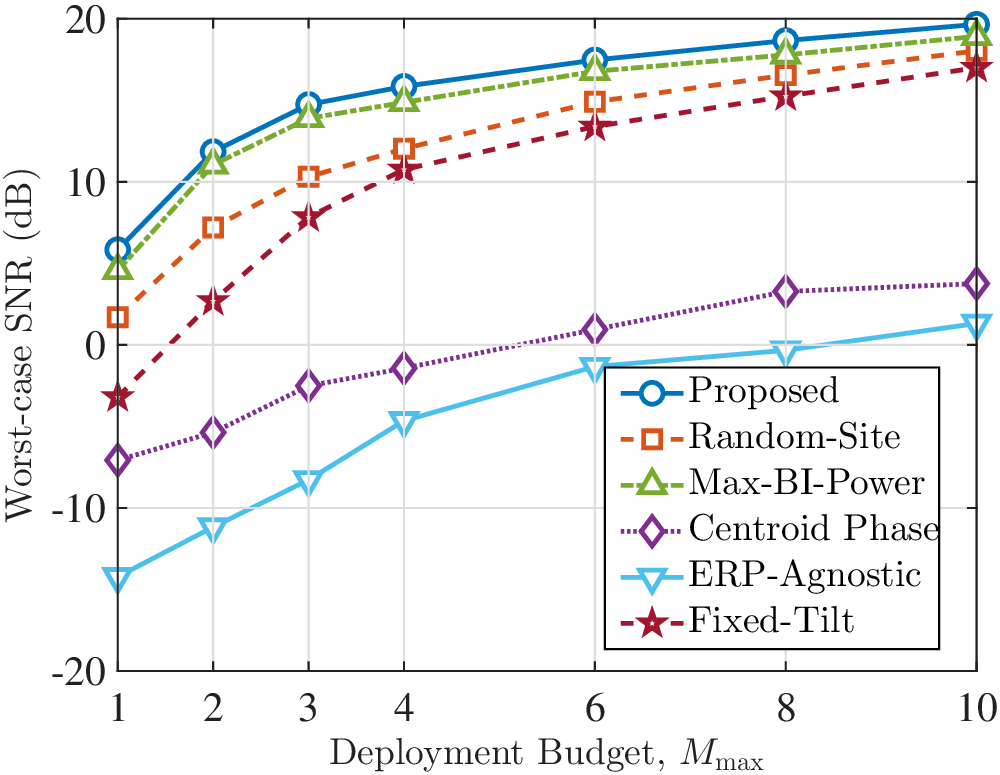}
		\caption{Worst-case SNR versus the deployment budget $M_{\max}$.}
		\label{fig:budget_comparison}
	\end{figure}
	Fig.~\ref{fig:budget_comparison} shows the worst-case SNR versus deployment budget $M_{\max}$ for the proposed scheme and five baselines. The No-IRS baseline is omitted because its SNR remains below the plotted range. The proposed scheme improves the worst-case SNR from $5.84$~dB at $M_{\max}=1$ to $17.47$~dB at $M_{\max}=6$ and $19.66$~dB at $M_{\max}=10$, while outperforming all five baselines. The increase reaches $11.63$~dB from $M_{\max}=1$ to $M_{\max}=6$, compared with only $2.19$~dB from $M_{\max}=6$ to $M_{\max}=10$. Max-BI-Power remains close to the proposed scheme and ranks candidate sites according to the BS-to-IRS incident power $\beta_{\mathrm{BI},m}G_{\mathrm B}(\mathbf q_m)$. The proposed scheme jointly determines site selection, IRS orientations, and phase shifts under the regional worst-case SNR objective. The gain over Fixed-Tilt demonstrates the value of IRS orientation optimization. ERP-Agnostic quantifies the loss caused by designing with an isotropic IRS element model and evaluating with the actual element pattern. The substantially lower SNR of Centroid Phase shows that co-phasing toward a single centroid direction cannot maintain high reflected gain over the entire covered direction range. All plotted schemes improve as $M_{\max}$ increases, while the marginal gain decreases beyond $M_{\max}=6$. The optimal value of problem~\textrm{(P1)} is nondecreasing with $M_{\max}$ because increasing the deployment budget enlarges the feasible set, whereas the diminishing gain in the obtained results reflects the reduced benefit of additional IRSs in the considered setup.

	\begin{figure}[t!]
		\centering
		\includegraphics[width=3.4   in]{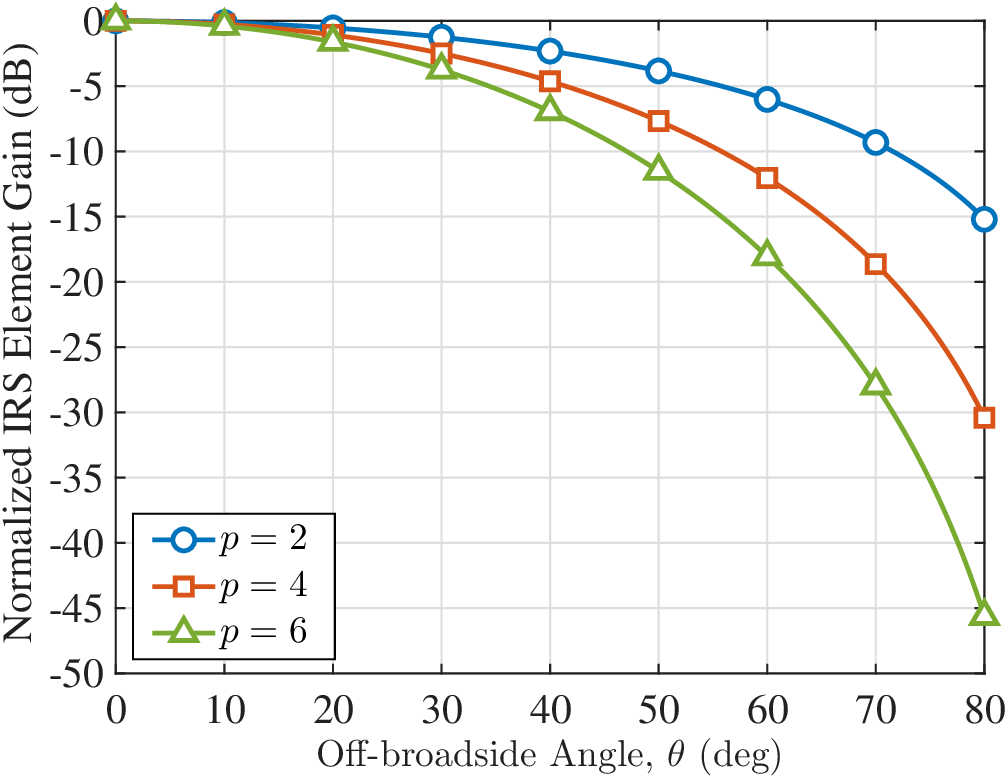}
		\caption{Normalized IRS element patterns for different $p_I$.}
		\label{fig:erp}
	\end{figure}
		
	Fig.~\ref{fig:erp} plots the normalized IRS element pattern in~\eqref{eq:irs_erp} versus the angle $\omega$ from the panel normal for $p_I=2$, $4$, and $6$. The $3$~dB half-width narrows from $45.0^\circ$ to $32.8^\circ$ and $27.0^\circ$, while the $10$~dB half-width decreases from $71.6^\circ$ to $55.8^\circ$ and $47.2^\circ$ as $p_I$ increases. At $\omega=60^\circ$, the normalized gain drops by $6.02$~dB, $12.04$~dB, and $18.06$~dB in proportion to $p_I$. The cascaded link in~\eqref{eq:cascaded_channel_expand} contains both $G_{\mathrm I}(\theta_m^{\mathrm{in}})$ and $G_{\mathrm I}(\theta_{m,u}^{\mathrm{out}})$, so the incident and outgoing angular losses add in dB. For $p_I=2$, a panel with both angles equal to $45^\circ$ incurs a $6.02$~dB loss, while the same geometry incurs $18.06$~dB for $p_I=6$. Increasing $p_I$ also raises the broadside gain $G_{\mathrm I}^{\max}=2(p_I+1)$ from $7.78$~dBi to $10.00$~dBi and $11.46$~dBi, but the absolute gain for $p_I=6$ falls below that for $p_I=2$ beyond $\omega=36.0^\circ$. A more directive element therefore provides higher gain only when the incident and outgoing angles remain sufficiently close to broadside, making IRS orientation increasingly important as $p_I$ grows.
	
	\begin{figure}[t!]
		\centering
		\includegraphics[width=3.4  in]{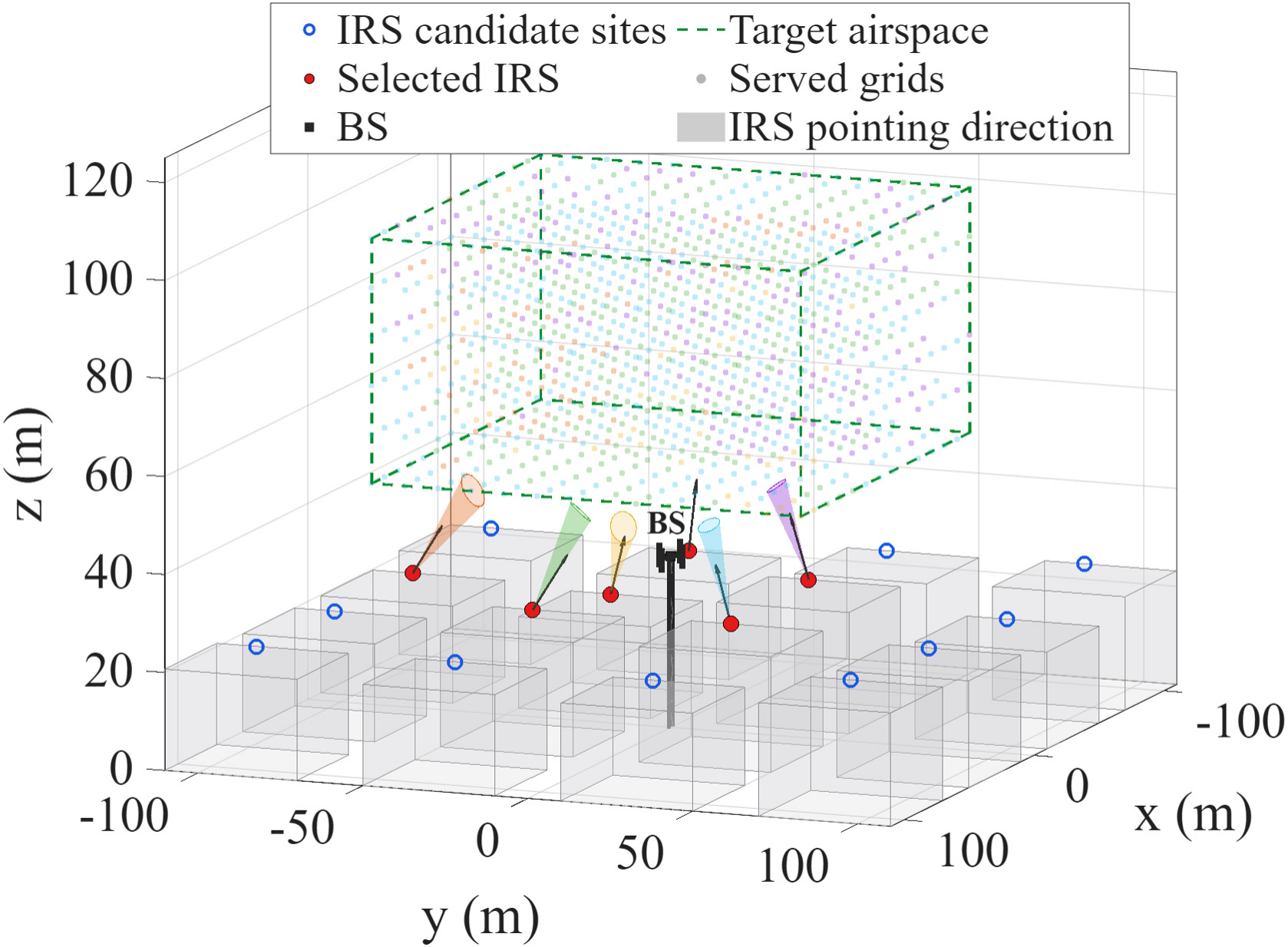}
		\caption{Optimized 3D IRS deployment.}
		\label{fig:site_selection}
	\end{figure}
	
	Fig.~\ref{fig:site_selection} illustrates the optimized 3D deployment for $M_{\max}=6$. Hollow blue markers denote the candidate rooftop sites, red markers denote the six selected IRSs, and the dark arrows show the optimized normal of each selected IRS. Each displayed location is colored according to the selected IRS that provides the largest reflected power. The translucent cones connect IRSs to the centroids of the displayed locations for which the corresponding IRSs provide the largest reflected power and visualize the associated reflection directions. The legend entry IRS pointing direction denotes these cone surfaces. Five selected IRSs have corresponding displayed locations and translucent cones. The sixth IRS directly contributes to the coherently combined received signal over the target airspace. The selected IRSs are distributed on different sides of the BS and have different optimized normals, while the interleaved color distribution shows that the IRS providing the largest reflected power varies across the target airspace. The direction span of the whole target airspace seen from a rooftop candidate ranges from $76^\circ$ to $143^\circ$ across the candidate set. For $N_m=100$ and $\bar{\kappa}=0.9$,~\eqref{eq:regional_span_guarantee} gives a sufficient direction span of $4.07^\circ$. This span is much smaller than the full-region spans observed at all candidate sites and places the prescribed normalized-gain guarantee in a narrow directional range. The optimized deployment employs spatially separated rooftop sites and different IRS orientations to provide complementary reflected links over different parts of the target region.
	\begin{figure}[t]
		\centering
		\setlength{\tabcolsep}{1pt}
		\renewcommand{\arraystretch}{0.8}
		\begin{tabular}{ccc}
			\small No-IRS, $z=55$\,m &
			\small No-IRS, $z=95$\,m \\[2pt]
			\includegraphics[
			width=0.23\textwidth,
			trim=0 0 0 28,
			clip
			]{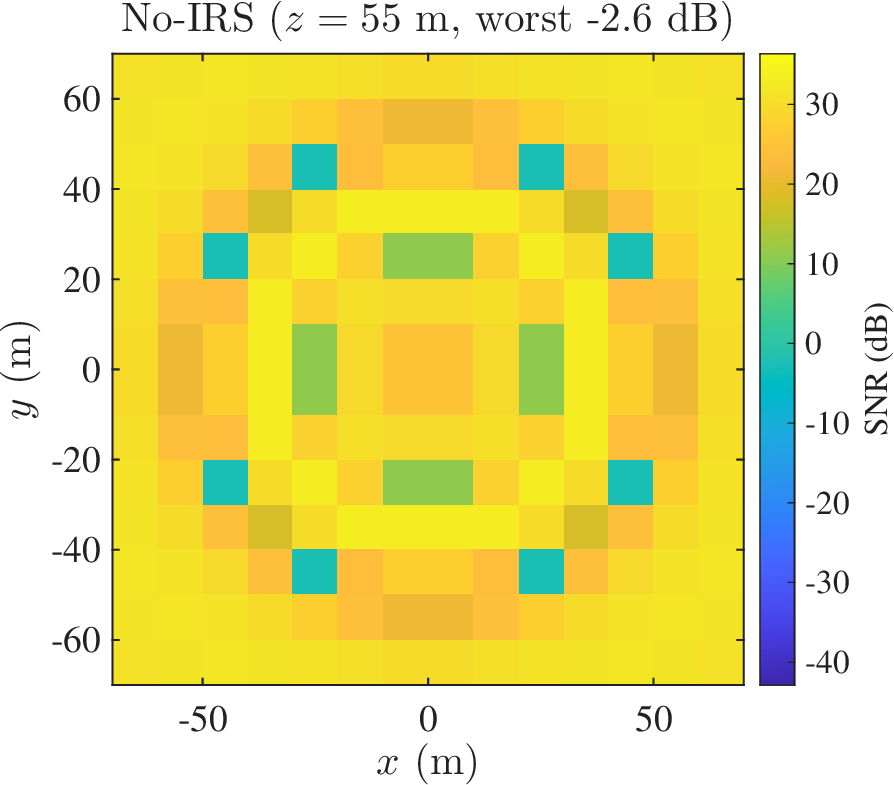} &
			\includegraphics[
			width=0.23\textwidth,
			trim=0 0 0 28,
			clip
			]{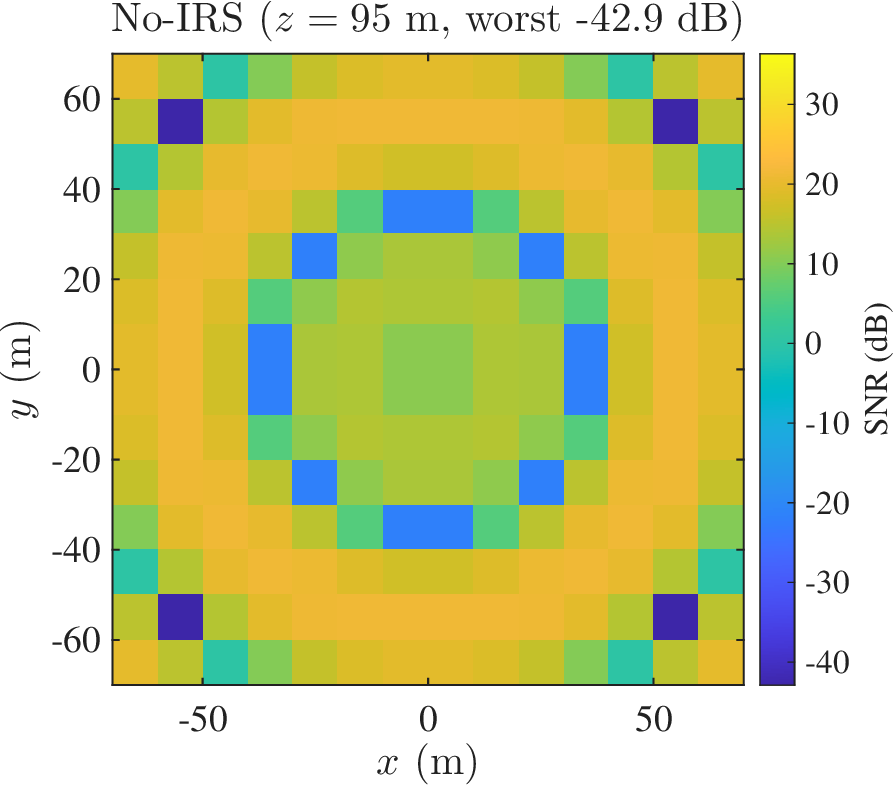} \\[4pt]
			\small Proposed, $z=55$\,m &
			\small Proposed, $z=95$\,m \\[2pt]
			\includegraphics[
			width=0.23\textwidth,
			trim=0 0 0 28,
			clip
			]{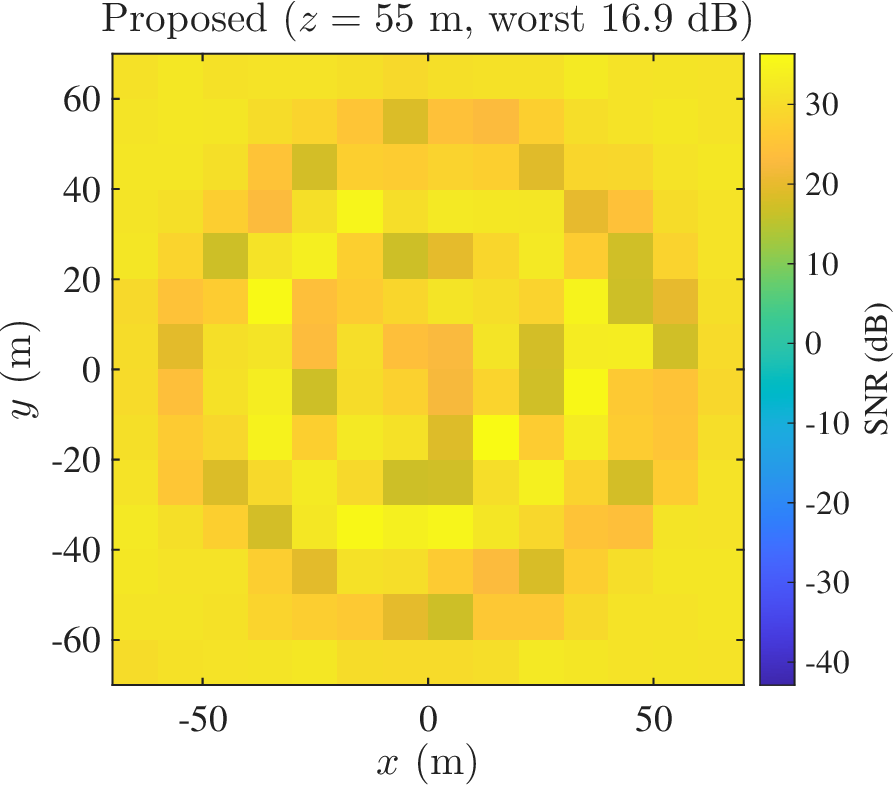} &
			\includegraphics[
			width=0.23\textwidth,
			trim=0 0 0 28,
			clip
			]{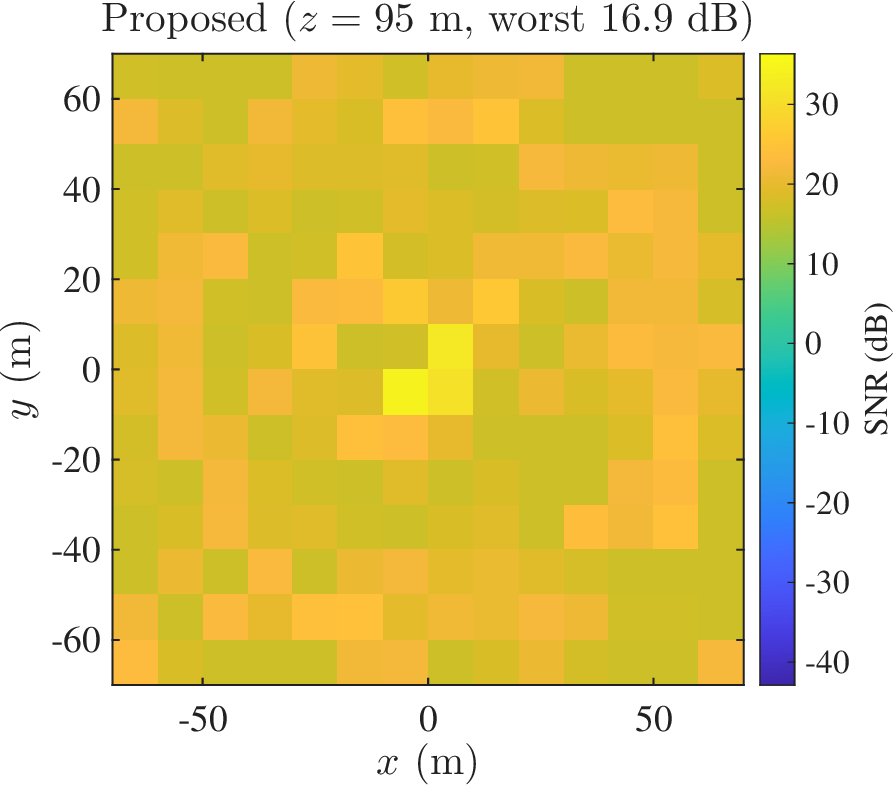}
		\end{tabular}
		\caption{Horizontal SNR heatmaps at altitudes $z=55$\,m and $z=95$\,m.} 
		\label{fig:heatmap}
	\end{figure}
	
	Fig.~\ref{fig:heatmap} presents horizontal SNR heatmaps at $z=55$~m and $z=95$~m. The top row shows the No-IRS baseline, and the bottom row the proposed scheme with $M_{\max}=6$. Under the No-IRS baseline, annular low-SNR regions arise from elevation-dependent nulls of the fixed-downtilt BS array, and the null radius increases with altitude. The worst SNR is $-2.61$~dB at $z=55$~m and decreases to $-42.88$~dB at $z=95$~m, where the altitude slice intersects a deep array-factor null. The proposed IRS deployment raises the worst SNR at both altitudes to about $17.4$~dB. The improvement is especially pronounced at $z=95$~m, where the minimum SNR increases by about $60.3$~dB compared with about $20.0$~dB at $z=55$~m. The $z=95$~m improvement shows that reflected links provide their largest benefit near severe direct-link nulls. The low-SNR rings are substantially compensated by the reflected links, and the SNR distribution becomes more uniform over each horizontal slice. The similar minimum SNR values at the two altitudes are consistent with the max-min objective, which directs the joint site, orientation, and phase-shift design toward weak locations at different heights.

	\begin{figure}[t!]
		\centering
		\includegraphics[width=3.4  in]{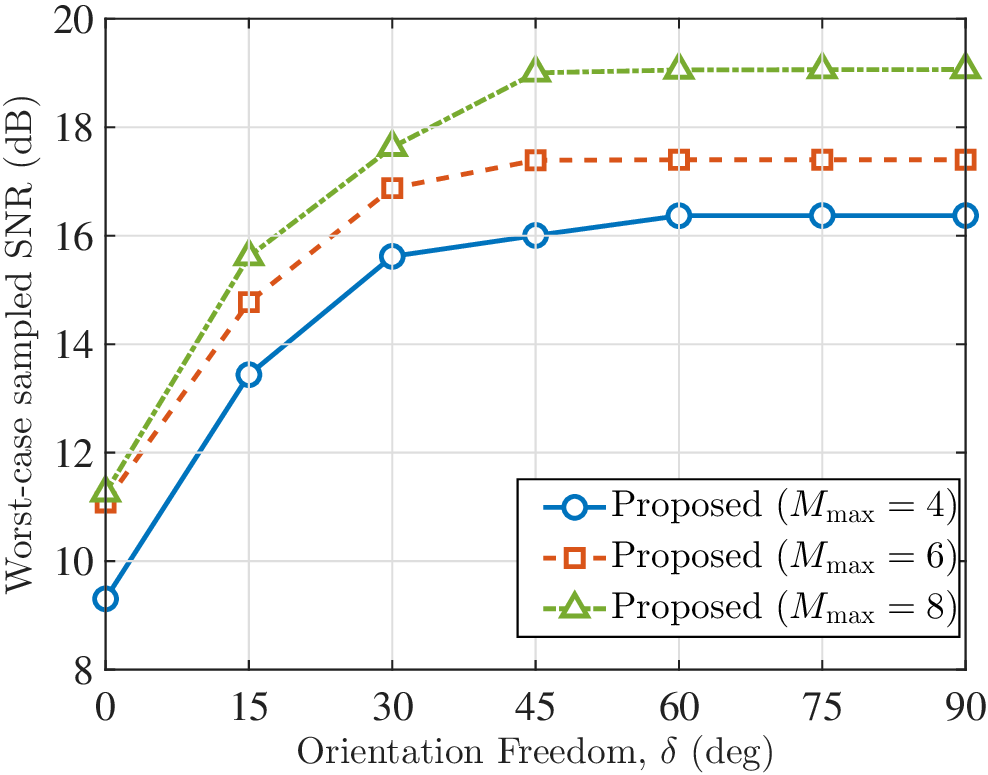}
		\caption{Worst-case SNR versus the orientation freedom $\delta$ for $M_{\max}$.}
		\label{fig:orientation}
	\end{figure}

	In Fig.~\ref{fig:orientation}, we plot the worst-case SNR versus the orientation freedom $\delta$ for the deployment budgets $M_{\max}=4$, $6$, and $8$. Here $\delta$ is the half-angle of the feasible cone imposed on each IRS normal, so that the feasible orientation set in~\eqref{eq:orientation_feasible_set} becomes $\mathcal T_m=\{(\theta_m,\phi_m)\mid\theta_m\in[0^\circ,\delta],\ \phi_m\in[0^\circ,360^\circ)\}$. The worst-case SNR increases with $\delta$ because every orientation feasible under a smaller cone remains feasible after the cone is enlarged. For $M_{\max}=8$, the worst-case SNR rises from $11.26$~dB at $\delta=0^\circ$ to $19.06$~dB at $\delta=90^\circ$, a gain of $7.80$~dB from orientation optimization. The increase is concentrated at small $\delta$, since widening the cone from $0^\circ$ to $45^\circ$ recovers $7.74$~dB of this gain. Increasing $\delta$ further from $45^\circ$ to $90^\circ$ adds only $0.06$~dB, while the mean optimized elevation reaches $44.4^\circ$. The small additional gain beyond $45^\circ$ indicates that most of the useful orientation adjustment is achieved within moderate inclination angles. The gap among the three deployment budgets widens as $\delta$ grows, reaching $1.96$~dB between $M_{\max}=8$ and $M_{\max}=4$ at $\delta=0^\circ$ and $2.69$~dB at $\delta=90^\circ$. At $\delta=0^\circ$, every IRS normal is fixed to the vertical, which restricts all panels to the same mounting orientation. A larger cone allows the selected IRSs to direct their responses toward complementary parts of the airspace and improve the weakest locations.
	
	\section{Conclusion}\label{sec:conclusion}
	In this paper, we studied rooftop IRS deployment for improving the worst-case SNR over a 3D low-altitude airspace served by an existing fixed-downtilt BS. A max-min SNR problem jointly selected rooftop sites and optimized IRS orientations and phase shifts under a deployment budget, with optimized parameters fixed after deployment. The continuous target region was represented by a finite location set augmented near predicted BS array-factor null heights for optimization. We proposed a mixed-integer AO algorithm that retains the binary site-selection variables while generating a nondecreasing and convergent sequence of worst-case SNR values. The analytical results revealed complementary array-size effects on rooftop illumination and directional characteristics. Larger BS arrays narrow the rooftop-height interval receiving main-lobe illumination, whereas larger IRS arrays narrow the sufficient outgoing-direction span retaining a prescribed normalized array gain. Numerical results validated both analytical characterizations and showed that the proposed scheme achieved higher worst-case SNR than the considered benchmarks across deployment budgets. The numerical comparisons confirmed the importance of IRS orientation optimization and region-wide phase-shift design while accounting for the element radiation pattern.

	\bibliographystyle{IEEEtran}
	\bibliography{reference}    
	
\end{document}